\documentclass[10pt]{amsart}

\usepackage{amsmath,amsthm, graphicx, amssymb, amscd}
\usepackage{tikz-cd}
\usepackage{multicol}	
\usepackage{graphics,subfigure}

\usepackage{mathtools}

\usepackage{nomencl}
\makenomenclature
\makeglossary

\usepackage{amsfonts}
\usepackage{pxfonts}
\usepackage{psfrag}
\usepackage{stmaryrd}
\usepackage{mathabx}
\usepackage{textcomp}
\usepackage{geometry} 
\usepackage{xcolor} 
\usepackage{framed} 
\usepackage{hyperref}
\newdimen\tw \tw=\textwidth\advance\tw-2in
\definecolor{shadecolor}{gray}{0.8} 

\DeclareMathOperator*{\Lzero}{L}

\newtheorem{theorem}{Theorem}[section]

\newtheorem{corollary}[theorem]{Corollary}
\newtheorem{definition}[theorem]{Definition}
\newtheorem{example}[theorem]{Example}
\newtheorem{lemma}[theorem]{Lemma}

\newtheorem{property}[theorem]{Property}
\newtheorem{proposition}[theorem]{Proposition}

\newtheorem{remark}[theorem]{Remark}

\begin{document}	

\newcommand{\cmul}{\star}

\newcommand{\BBbr}{\llbracket\cdot\rrbracket}

\newcommand{\Lmark}{$\bullet$}	

\newcommand{\dmrjdel}[1]{}
\newcommand{\myfootnote}[1]{\footnote{#1}}

\newcommand{\HRule}{\rule{\linewidth}{0.5mm}}

\newcommand{\sze}[1]{ \mbox{\tiny{\textit {#1}}} }        

\renewcommand{\thesection}{\arabic{section}}

\newcommand{\puteps}[2][0.5]{\includegraphics[scale=#1]{#2}}
\newcommand{\myputeps}[3]{\raisebox{#1 mm}{\puteps[#2]{#3} }}
\newcommand{\Bmyputeps}[3]{\left(\raisebox{#1 mm}{\puteps[#2]{#3} }\right)}
\newcommand{\Cmyputeps}[3]{\left\{\raisebox{#1 mm}{\puteps[#2]{#3} }\right\}}

\newcommand{\pf}{\noindent {\sc Proof. }}
\newcommand{\myss}[1]{\noindent\underline{\emph{#1}}:\quad}
\newcommand{\explain}[1]{\mbox{(#1)}}
\newcommand{\ds}[1]{\displaystyle{#1}}

\newcommand{\ha}{\frac{1}{2}}
\newcommand{\haA}{\frac{1}{2}a^{-1}}
\newcommand{\subha}{\mbox{\textonehalf}}
\newcommand{\qu}{\frac{1}{4}}

\newcommand{\Lz}[1]{\Lzero_{{#1}=0}}	

\newcommand{\aut}{\mathsf{aut}}
\newcommand{\bij}{\stackrel{\sim}{\longrightarrow}}
\newcommand{\ip}[2] { \left\langle{#1,#2}\right\rangle}
\newcommand{\Fi}{F_\mathsf{int}}
\newcommand{\Si}{S_\mathsf{int}}
\newcommand{\ai}{a_\mathsf{int}}
\newcommand{\rar}{\rightarrow}  
\newcommand{\boxt}{\cdot}  
\newcommand{\ot}{\otimes}
\newcommand{\oc}{\circ}
\newcommand{\val}{\mathsf{val}}

\newcommand{\pc}{\partial_{\mathsf{\sfc}}}
\newcommand{\pe}{\partial_{\mathsf{\sfe}}}
\newcommand{\ps}{\partial_{\mathsf{\sfs}}}
\newcommand{\pr}{\partial_{\mathsf{\sr}}}
\newcommand{\pp}{\partial_{\mathsf{\sfp}}}
\newcommand{\pv}{\partial_{\mathsf{\sv}}}
\newcommand{\pvk}[1]{\partial_{\mathsf{\sv}_{#1}}}

\newcommand{\prtl}[2]{ \frac{\partial #1}{\partial #2}}
\newcommand{\dprtl}[2]{ \frac{d #1}{d #2}}
\newcommand{\pd}[1]{\,\partial_{#1}}

\newcommand{\KTR}{\mathbb{K}}
\newcommand{\AFT}{\mathbb{K}_a}
\newcommand{\tAFT}[1]{\mathcal{X}_{#1}}
 \newcommand{\nFT}[1]{\mathcal{F}_{#1}^{\mathsf{alg}}}

\newcommand{\sfone}{\mathsf{1}}

\newcommand{\wt}[1]{\om_{#1}}
\newcommand{\gens}[2]{\llbracket #1, \wt{#2} \rrbracket}
\newcommand{\gensb}[2]{\llbracket #1, #2 \rrbracket}
\newcommand{\rppr}[2]{\ds{\frac{#1 \partial}{\partial #1} #2}}
\newcommand{\setP}[3]{\cP^#1_{#2,#3}}
\newcommand{\pM}{\ds{\frac{\partial\cM}{\partial\sfp}}}
\newcommand{\ol}[1]{\overline{#1}}

\newcommand{\fnspace}{\cV}

\newcommand{\prob}{\mathsf{prob}}
\newcommand{\delJ}[1]{\frac{\de}{\de J(\bx_{#1})}}
\newcommand{\Gcon}{G_{\mathsf{conn}}}

\newcommand{\Gtwoc}{G_{2\mathsf{-conn}}}
\newcommand{\Gtwolc}{\mathcal{G}^{\ell,2\mathsf{-conn}}_\cF}

\newcommand{\etwo}{\stackrel{\ast}{=}}

\newcommand{\mi}[1]{{#1}^{\,\mbox{\tiny\textrm{[-1]}}}} 
\newcommand{\dmi}[1]{\stackrel{\centerdot}{#1}^{\,\mbox{\tiny[-1]}}} 
\newcommand{\diff}[1]{\stackrel{\centerdot}{#1}}  

\newcommand{\SOe}[2]{(\mbox{\textsf{#1}}, \;{#2},\;\mbox{\textit{exp.}})}
\newcommand{\SOo}[2]{(\mbox{\textsf{#1}}, \;{#2},\;\mbox{\textit{ord.}})}

\newcommand{\al}{\alpha}
\newcommand{\be}{\beta}
\newcommand{\de}{\delta}
\newcommand{\varep}{\varepsilon}
\newcommand{\blam}{\boldsymbol{\lambda}}
\newcommand{\bnu}{\boldsymbol{\nu}}
\newcommand{\Ga}{\Gamma}
\newcommand{\hGa}{\hat{\Ga}}
\newcommand{\lam}{\lambda}
\newcommand{\Lam}{\Lambda}
\newcommand{\Ups}{\Upsilon}
\newcommand{\om}{\omega}
\newcommand{\Om}{\Omega}
\newcommand{\bphi}{\boldsymbol{\phi}}
\newcommand{\si}{\sigma}
\newcommand{\vphi}{\varphi}
\newcommand{\sqa}{\sqrt{a}}

\newcommand{\bul}{\bullet}
\newcommand{\bzero}{\mathbf{0}} 

\newcommand{\LT}{\mathsf{L}\,}
\newcommand{\fLT}{\mathbb{L}\,}
\newcommand{\cLT}{\bbL^\mathsf{comb}}

\newcommand{\FT}{\mathsf{F}\,}
\newcommand{\fFT}{\mathbb{F}\,}
\newcommand{\cFT}{\bbF^\mathsf{comb}}

\newcommand{\Gt}{{K}}	
\newcommand{\bbA}{\mathbb{A}}
\newcommand{\cA}{\mathcal{A}}
\newcommand{\cN}{\mathcal{N}}
\newcommand{\cO}{\mathcal{O}}
\newcommand{\ba}{\mathbf{a}}
\newcommand{\fa}{\mathfrak{a}}
\newcommand{\sa}{\mathsf{a}}
\newcommand{\bA}{\mathbf{A}}
\newcommand{\bbB}{\mathbb{F}}	
\newcommand{\mbbB}{\mathbb{B}}
\newcommand{\cB}{\mathcal{B}}
\newcommand{\bb}{\mathbf{b}}
\newcommand{\fb}{\mathfrak{b}}
\newcommand{\sfb}{\mathsf{b}}
\newcommand{\cC}{\mathcal{C}}
\newcommand{\bc}{\mathbf{c}}
\newcommand{\sfc}{\mathsf{c}}
\newcommand{\bbC}{\mathbb{C}}
\newcommand{\cD}{\mathcal{D}}
\newcommand{\fD}{\mathfrak{D}}
\newcommand{\bd}{\mathbf{d}}
\newcommand{\cE}{\mathcal{E}}
\newcommand{\sfe}{\mathsf{e}}
\newcommand{\bbF}{\mathbb{F}}
\newcommand{\cF}{\mathcal{F}}
\newcommand{\sfF}{\mathsf{F}}
\newcommand{\hf}{\hat{f}}
\newcommand{\hg}{\hat{g}}
\newcommand{\cG}{\mathcal{G}}
\newcommand{\g}{{\sf g}}
\newcommand{\fg}{\mathfrak{g}}
\newcommand{\h}{{\mathsf h}}
\newcommand{\fh}{\mathfrak{h}}
\newcommand{\cH}{\mathcal{H}}
\newcommand{\bI}{\mathbf{I}}
\newcommand{\sI}{\mathsf{I}}
\newcommand{\bk}{\mathbf{k}}
\newcommand{\fl}{\mathfrak{l}}
\newcommand{\bbL}{\mathbb{L}}
\newcommand{\cL}{\mathcal{L}}
\newcommand{\sL}{\mathsf{L}}
\newcommand{\bM}{\mathbf{M}}
\newcommand{\fm}{\mathfrak{m}}
\newcommand{\cM}{\mathcal{M}}
\newcommand{\bbN}{\mathbb{N}}
\newcommand{\so}{\mathsf{o}}
\newcommand{\cP}{\mathcal{P}}
\newcommand{\sfp}{\mathsf{p}}
\newcommand{\fp}{\mathfrak{p}}
\newcommand{\bp}{\mathbf{p}}
\newcommand{\sfq}{\mathsf{q}}
\newcommand{\cQ}{\mathcal{Q}}
\newcommand{\bbQ}{\mathbb{Q}}
\newcommand{\bbR}{\mathbb{R}}
\newcommand{\sR}{\mathsf{R}}
\newcommand{\cR}{\mathcal{R}}
\newcommand{\fr}{\mathfrak{r}}
\newcommand{\sr}{\mathsf{r}}
\newcommand{\cS}{\mathcal{S}}
\newcommand{\fS}{\mathfrak{S}}
\newcommand{\fs}{\mathfrak{s}}
\newcommand{\sS}{\mathsf{S}}
\newcommand{\sfs}{\mathsf{s}}
\newcommand{\bt}{\mathbf{t}}
\newcommand{\tr}{\mathfrak{t}}
\newcommand{\ft}{\mathfrak{t}}
\newcommand{\cT}{\mathcal{T}}
\newcommand{\sT}{\mathsf{T}}
\newcommand{\cU}{\mathcal{U}}
\newcommand{\cV}{\mathcal{V}}
\newcommand{\su}{\mathsf{u}}
\newcommand{\bu}{\mathbf{u}}
\newcommand{\sv}{\mathsf{v}}
\newcommand{\cW}{\mathcal{W}}
\newcommand{\bw}{\mathbf{w}}
\newcommand{\bx}{\mathbf{x}}
\newcommand{\hx}{\hat{x}}
\newcommand{\cX}{\mathcal{X}}
\newcommand{\by}{\mathbf{y}}
\newcommand{\hy}{\hat{y}}
\newcommand{\bz}{\mathbf{z}}
\newcommand{\cZ}{\mathcal{Z}}    

\newcommand{\tcb}[1]{\textcolor{blue}{#1}}


\title[Towards a more algebraic footing for quantum field theory]{Towards a more algebraic footing for quantum field theory}


\author[D.~M.~Jackson]{David~M.~Jackson$^*$}
\author[A. Kempf]{Achim Kempf$^\dagger$}
\author[A.H. Morales]{Alejandro H. Morales$^\ddagger$}

\thanks{
${\hspace{-1ex}}^*$Department of Combinatorics and Optimization, 
                                                University of Waterloo, Waterloo, Ontario, Canada;  \\
${\hspace{.35cm}}$ \texttt{dmjackson@math.uwaterloo.ca}}

\thanks{
${\hspace{-1ex}}^\dagger$Departments of Applied Mathematics and Physics, 
					University of Waterloo, Waterloo, Ontario, Canada;   \\
${\hspace{.35cm}}$ \texttt{akempf@perimeterinstitute.ca}}

\thanks{
${\hspace{-1ex}}^\ddagger$Department of Mathematics and Statistics, 
                                                University of Massachusetts, Amherst, MA, USA;  \\ %
${\hspace{.35cm}}$ \texttt{ahmorales@math.umass.edu}}

\date{\underline{\textbf{\today}}}

\begin{abstract}
The predictions of the standard model of particle physics are highly successful in spite of the fact that several parts of the underlying quantum field theoretical framework are  analytically problematic.
Indeed, it has long been suggested, by Einstein, Schr\"odinger and others, that analytic problems in the formulation of fundamental laws could be overcome by reformulating these laws without reliance on analytic methods namely, for example, algebraically.  
In this spirit, we  focus here on the analytic ill-definedness of the quantum field theoretic Fourier and Legendre transforms of the generating series of Feynman graphs, including the path integral.  
To this end, we  develop here purely algebraic and combinatorial formulations of the Fourier and Legendre transforms, employing rings of formal power series. These are all-purpose transform methods, i.e., their applicability is not restricted to quantum field theory. When applied in quantum field theory to the generating functionals of Feynman graphs, the new transforms are well defined and thereby help explain the robustness and success of the predictions of perturbative quantum field theory in spite of analytic difficulties. 
Technically, we overcome here the problem of the possible divergence of the various generating series of Feynman graphs by constructing Fourier and Legendre transforms of formal power series that operate in a well defined way on the coefficients of the power series irrespective of whether or not these series converge. 
Our new methods could, therefore, provide new algebraic and combinatorial perspectives on quantum field theoretic structures that are conventionally thought of as analytic in nature, such as the occurrence of anomalies from the path integral measure. 

In comparison, the use of formal power series in QFT by Bogolubov, Hepp, Parasiuk and Zimmermann concerned a different kind of divergencies, namely the UV divergencies of loop integrals and their renormalization.

\end{abstract}

\newpage
\maketitle

 \newpage
 \tableofcontents
 \newpage

\parskip=1pt	

\pagestyle{myheadings}	



\section{Introduction}

Advocacy for the use of algebraic methods in physics can be traced back at least to Schr\"odinger and Einstein. For example, Penrose
\cite{P}
quotes Schr\"odinger (1950) ``The idea of a continuous range, so familiar to mathematicians in our days, is something quite exorbitant, an enormous extrapolation of what is accessible to us."
Einstein wrote \cite{E}   
``Quantum phenomena . . . must lead to an attempt to
find a purely algebraic theory for the description of reality." From this perspective, our aim here is to show that key structures and maps whose very definition pose persistent analytic problems in quantum field theory, such as the path integral and its transforms, are indeed better understood, and well defined, algebraically and combinatorially.   

\subsection{Power series whose role it is to represent discrete structures, instead of functions} 

We begin with the observation that, in mathematical physics, it is commonplace to consider functions and the power series obtained by their Taylor expansion as being equivalent and interchangeable within the power series' radius of convergence. While this practice is justified, it is important to point out that the utility of power series can extend far beyond providing series representations of functions. 

What we have in mind is that power series also provide powerful tools as generating series of discrete structures that can be described combinatorially and or algebraically (e.g., \cite{GJ}). This use of power series is logically disconnected from the use of power series to represent functions. Power series used as generating series for discrete structures need not represent any function to be useful, i.e., they need not converge.  
This is because power series used as generating series for discrete structures encode all physical information in their coefficients. There is no need to be able to sum up the terms of a generating series to obtain a function. Since these series need not represent functions, they are also called formal power series. Algebraically, they are elements of rings of formal power series.  

A natural question to ask, therefore, is whether there are occurrences of power series in mathematical physics where it is not the role of the power series to represent functions but where it is their role to represent generating series for discrete structures. 
In such a circumstance, any expectation that these power series converge so that they can represent functions would be misplaced. In particular, it would also be misplaced to further assume that these power series represent functions that are related to each other by analytic transforms, such as the conventional Legendre or Fourier transforms. 

\subsection{Algebraic and combinatorial Legendre and Fourier transforms}

In circumstances where power series represent discrete structures, if maps between them are called for that are analogous to say Legendre or Fourier transforms, then these maps should be algebraically or combinatorially defined as maps between the coefficients of the pre-image power series and the image power series. Such combinatorially or algebraically defined maps between power series can be well defined whether or not the power series in question represent functions, i.e., these combinatorially or algebraically defined maps can be well defined irrespective of whether the power series possesses a finite radius of convergence.   

As is well known, in quantum field theory, the generating series of all graphs and of connected graphs and the series representing the quantum effective action generally are divergent.  In the present paper, we argue that the quantum field theoretical action, the generating functionals of all Feynman graphs and of all connected Feynman graphs as well as the quantum effective action are series that are best understood not as functions but as generating series for the discrete algebraic and combinatorial structures of Feynman graphs and that, as such, they need not converge. 

Correspondingly, we argue that the maps that relate these generating series, i.e., the quantum field theoretical path integral, Fourier transform, logarithm and Legendre transform ought to be re-derived not as analytic functions but as maps between the coefficient sets of these generating series. So far, these tasks have been partially accomplished.  

First, the fact that the logarithm has a combinatorial analog that reduces a generating series of all graphs to a generating series of only the connected graphs has long been known. Further, in previous work, we showed that also the Legendre transform can be combinatorially and algebraically re-defined as a map between the coefficient sets of power series - while retaining the Legendre transform's defining equations but now being applicable also to power series that are divergent. 
In quantum field theory, the algebraically and combinatorially defined Legendre transforms now yield a well-defined map of the series representing the action into the generating series of tree graphs, as well as a well-defined map between the generating series of connected graphs and the quantum effective action. These maps are well defined in spite of the fact that in theories with interactions, these generating series are generally divergent.   

In the present paper, we pick up the remaining challenge to also redefine the Fourier transform algebraically and combinatorially, i.e., we will re-define the Fourier transform as a map between the coefficient sets of generating series that is valid irrespective of the convergence properties of these power series. 

The significance of reconstructing the Fourier transform algebraically and combinatorially derives from the fact that the quantum field theoretic path integral itself is a Fourier transform
\begin{equation} 
Z[J] = 
\int_\cD e^{i\, S[\phi]} e^{i\int_{\bbR^4} \phi(\bx) J(\bx) d^4\bx} D[\phi],
\end{equation}
as we will discuss in more detail below.

While the path integral is notoriously hard to define analytically we argue here that it in fact need not be defined analytically. 
Technically, re-defining the Fourier transform algebraically and combinatorially is well-defined and robust to analytic issues because it is well defined as an operation that maps formal power series into formal power series. As we will show, this map is well-defined irrespective of analytic considerations because in order to calculate any coefficient of the image power series only a {\em finite number of coefficients} of the pre-image power series need to be known, i.e., we have so-called local finiteness. In prior work \cite{JKM1}, we showed that this program can be carried out for the Legendre transform. In particular, we proved combinatorial/algebraic analogues of properties such as the involutive property and we obtain explicit formulas and in some instances closed expressions \cite[Section 5]{JKM1} for the Legendre transform.  However, also for our algebraic/combinatorial Fourier transform  here we get explicit formulas (see  Appendix~\ref{S:appendix_explicit_examples}). In addition, our previous work \cite{JKM2,JKM3} lead to all-purpose integration methods. It would be of interest to explore to what extent these new methods can be applied to the evaluation of Feynman diagrams. Note that we do not address here the integral evaluation of each Feynman diagram. However, related to the present work, we also developed powerful exact new algebraic methods for  integration and integral transforms, see \cite{JKM3,JiTK}, which have recently been built into {\tt Maple} and which are standard as of {\tt Maple2019} \cite{Maple}. These integration methods may be  applied to the evaluation of the integrals of individual Feynman diagrams. For example, the new integration methods have been successfully applied to evaluate nontrivial integrals, such as the {\em Borwein integrals}, in \cite{JiTK}.

\subsection{Possible applications}

For the longer term, it is our hope that the new algebraic and combinatorial Fourier and Legendre transforms may provide a new perspective as well new tools with which to investigate the origin of structures in quantum field theory which are usually thought of as being analytic in nature. For example, the algebraically and combinatorially generalized Fourier transforms that we present here may provide new insights into the origin of ghost fields from gauge fixing or into the origin of anomalies from the path integral measure. Also, in quantum cosmology, certain path integrals of mini superspace models are of key importance concerning the challenge of developing a theory for the initial conditions of the universe, such as the Hawking Hartle no-boundary proposal. These path integrals are, analytically, of the type of highly oscillatory Fourier integrals and their exact analytic definition and evaluation can be ambiguous, which has led to much discussion in the literature, see, e.g., \cite{H1,H2,HarH,FLT1,FLT2,DHalHarHeJan}. The new perspective and the new algebraic Fourier transform that we present here may be of help in this context. For example, the simplifying assumptions made in these mini superspace models lead to expressions with negative powers in the exponent that lie outside of the ring of formal power series in which the full path integral lives. This may be a hint as to the origin of ambiguities that have been discussed in the literature. 

Further, the algebraic and combinatorial Fourier and Legendre transforms that we construct here are new mathematical tools which are not necessarily tied to perturbative quantum field theory. The new transformations could be used in any context where the Legendre or Fourier transforms of series, such as generating series of any kind, are needed, even when these series do not converge. 

\subsection{Historical background}

Also very interesting but for purposes different from ours, algebraic methods associated with formal power series in QFT have been used in the work initiated by Bogolubov, Hepp, Parasiuk and Zimmermann on the renormalization of the ultraviolet divergences of quantum field theory. The BPHZ theorem uses the fact that for any fixed UV cutoff, formal power series in the bare masses and couplings can be re-expressed as formal power series in the renormalized masses and couplings. The theorem then states that the coefficients of the resulting perturbative formal power series are finite order by order as the UV cutoff is being removed. This is a statement about the \it finiteness of the coefficients \rm of formal power series in QFT. For recent developments following up on BPHZ, see, in particular, the work by Kreimer et. al. \cite{KK} on renormalization Hopf algebras. See also \cite{Ab,Bo,Bo2,Y} and  \cite[Chap. 3]{Z} for a comprehensive exposition.

Complementing these results, our work here is concerned with the \it finiteness of the number of coefficients  needed \rm to perform general Fourier and Legendre transforms, in QFT and also beyond QFT.


\subsection{The path integral and its associated Fourier and Legendre transforms} 
We begin with a brief review of the analytic difficulties that we aim to address. To this end, let $\phi$ stand for a field and let $S[\phi]$ be the field's action. The path integral of $e^{iS[\phi]}$ is the generating functional of Feynman graphs, $Z[J]$: 
\begin{equation} \label{equation: definition Z}
Z[J] = 
\int_\cD e^{i\, S[\phi]} e^{i\int_{\bbR^4} \phi(\bx) J(\bx) d^4\bx} D[\phi]
\end{equation}
Here, $J$ denotes a Schwinger source field. We notice that $Z[J]$ is of the form of a Fourier transform in continuously infinitely many variables of the exponentiated action, $e^{iS[\phi]}$. With the Lorentzian signature of spacetime, this Fourier transform is generally ill defined. One usually proceeds, nevertheless, by next taking the logarithm of $Z[J]$ to obtain the generating functional of the connected Feynman graphs, $W[J]=-i\log(Z[J])$. This particular step is known to be combinatorial in nature and it is, therefore, robust against analytic issues. Finally, the Legendre transform of $W[J]$ yields the quantum effective action, $\Gamma[\varphi]$. This step is analytically ill defined because for the analytic Legendre transform to be applicable to $W[J]$, $W[J]$ would need to be a well defined functional (in the form of a convergent power series) that is convex. In general, however, neither convexity nor even convergence of $W[J]$ can be assumed. Disregarding the analytic issues, the overall picture is:  
\begin{equation}\label{e:ChainTranforms4c}
S  \stackrel{\exp}{\rightsquigarrow} e^{iS} \stackrel{\FT}{\rightsquigarrow}  Z
 \stackrel{\log}{\rightsquigarrow}  W  \stackrel{\LT}{\rightsquigarrow}  \Ga.
\end{equation}
Here, the exponentiation and logarithm maps are well defined while the Fourier and Legendre transforms, $\FT$ and $\LT$ respectively, are not well defined analytically.

\subsection{Preview of the paper}
The main focus in the present paper will be on extending the algebraic and combinatorial approach that, in earlier work \cite{JKM1}, we applied to the Legendre transform to the Fourier transform and, therefore, to the path integral. For simplicity, as for example in \cite{JKM1,Bo}, here we will work with the univariate case which means we only have one field on a $0$-dimensional space. 

In Section~\ref{S:EnBackG}, we will give an account of a correspondence between (algebraic) operations on the ring of formal power series  and combinatorial operations on sets of combinatorial configurations. This section is confined in generality strictly to the purposes of this paper and includes the well-known analogue of Lagrange's Inversion Theorem.

The main part of the paper, namely, Section~\ref{S:FormalFT1}, gives the construction of both an algebraic Fourier transform (Definition~\ref{D:fFTa}, cf. \cite[Sec. 3]{Bo2}) and a combinatorial Fourier transform (Theorem~\ref{T:gsZF}), together with combinatorial proofs that they do indeed possess the fundamental properties analogous to those of their classical counterpart. In particular, we show also that the combinatorial Fourier transform is a quasi-involution (Theorem~\ref{T:QuasiInvFT}) and that it satisfies a product-derivation property (Lemma~\ref{ProdDeriv}, \S~\ref{subsec:combproofproduct-derivation}). For completeness, in Section~\ref{S:FormalLT1} we also include and expand the construction of the algebraic and combinatorial Legendre transform from \cite{JKM1} (Theorem~\ref{L:TreeDiff}) and include the first proof that the combinatorial Legendre transform is a quasi-involution (Theorem~\ref{T:QuasiInvolLT}). Notable related results to this Legendre transform have appeared in \cite{Ba,KY,MY}. In Section~\ref{S:outlook}  we give a summary and outlook with further questions. In the Appendix we give explicit locally finite formulas and examples of the combinatorial Fourier transform.

\subsection{Prior work on the Legendre transform of formal power series}
In \cite{JKM0,JKM1} we already carried out part of the  program outlined above for the case of the Legendre transform. Let us briefly review the main points. 

In QFT, the Legendre transform arises, for example, as the map between the action and the generating functional of tree graphs and as the map between the quantum effective action and the generating functional of connected graphs, which is the generating functional of tree graphs for the Feynman rules consisting of 1 particle irreducible (1PI) graphs. 

Mathematically, the Legendre transform is normally defined, analytically, as follows. \begin{definition}[The (Analytic) Legendre Transform $\LT$]\label{D: AnalLagTrans}
Let $f(x)$ be a convex function of $x$. Let $y:= f'(x)$ and let $x = g(y)$ be its (unique) compositional inverse.
The \emph{(analytic) Legendre Transform},  $\LT f$, of $f$ is defined as a function of a variable $z$ by
\begin{equation*}\label{e:aLT}
(\LT f)(z) := -z\cdot g(z) + (f\oc g)(z).
\end{equation*}
\end{definition}
The Legendre transform possesses the following well-known properties.
\begin{lemma}\label{e:yLfg}
Under the conditions of Definition~\ref{D: AnalLagTrans}: 
\begin{enumerate}
\item[(a)] $\LT$ is a quasi-involution on this set of functions; that is $(\LT^2 f)(-x) = f(x).$
\item [(b)] If $\ds{z = f'(x)}$ then $\ds{x = - \frac{d}{dz} (\LT f)(z)}.$ 
\end{enumerate}
\end{lemma}
We notice that the analytically-defined Legendre transform requires the existence of the compositional inverse of the derivative of the function in question. In quantum field theory, however, the existence of a compositional inverse, for example, of $dW[J]/dJ$ or $d\Gamma[\varphi]/d\varphi$, cannot be assumed. 

As we showed in \cite{JKM1}, the Legendre transform can be generalized algebraically and combinatorially for $f$ given as formal power series while preserving the properties (a) and (b) above and yet without requiring the existence of a compositional inverse, nor is it even necessary that the power series $f$ possesses a finite radius of convergence. We give an example from \cite{JKM1} to illustrate this. For the action $F(x) = -x^2/2 + \sum_{n\geq 3} (n-1)!x^n$ that has a zero radius of convergence as a function, in \cite[Ex. 7]{JKM1} we showed that the combinatorial Legendre transform is the following power series
\[
T(y)\,:=\,(\mathbb{L}\, F)(y) \,=\, \frac{y^2}{2} + \sum_{n=3}^{\infty} \left(\sum_{n_3,n_4,\ldots,n_k } \frac{(n-2+\sum_{j=3}^k n_j)!}{\prod_{j=3}^k n_j!} \prod_{j=3}^k (j!)^{n_j} \right) \frac{y^n}{n!},
\]
where the internal finite sum is over all finite tuples $(n_3,n_4,\ldots,n_k)$ of nonnegative integers satisfying the relation $\sum_{j=3}^k (j-2)n_j = n-2$. Here, $T(y)$ is the generating series of the tree graphs generated by the Feynman rules of the action $F(x)$.

\subsection{Advantages of the combinatorial over the conventional analytic Legendre transform} 

How this generalization is possible can serve to illustrate the power of the use of formal power series. To this end, we begin with the analytic consideration. 
Let us consider an entire function, $g(x)$, that passes through the origin with a nonzero slope, for example, $g_{ex}(x):=\sin(x)$. Its analytically-defined compositional inverse, $f(y)$, is generally multi-valued, for example $f_{ex}(y):=\arcsin(y)$. But there exists an interval containing $0$ in which $f$ is single-valued. In the example, $f_{ex}$ maps the interval $[-1,1]$ into the interval $[-\pi/2,\pi/2]$. The radius of convergence, $r_c$ of the MacLaurin series of $f$ is bounded from above by the size of this interval. For example, $r_c=1$ for $f_{ex}$ because $f_{ex}$ can cover only one half oscillation of the sine function before the sine function becomes non-monotonic. 

Let us now discuss this situation from the perspective of formal power series. We begin by observing that the function given by the formal Maclaurin series of $f$ on its radius of convergence describes only one branch of the multi-valued function $g^{-1}$. It would appear, therefore, that the coefficients of the MacLaurin series of $f$ only possess information about this one branch of $g^{-1}$. 

In fact, however, the coefficients of the Maclaurin series of $f$ contain complete information about $g^{-1}$ and $g$ everywhere. This is because, by the Lagrange inversion theorem, any formal power series whose constant term is zero and whose linear term is nonzero (as is the case here) possesses a compositional inverse in the form of a formal power series. In the example above, this inverse is the Maclaurin series of $g(x)=\sin(x)$. The radius of convergence of this Maclaurin series is infinite, and this implies that we obtain the entire function $g(x)$. In the example above, this includes \it all \rm the oscillations of the sine function $g_{ex}(x)$, i.e., it also includes all the branches of its multi-valued inverse. 

This robustness of the compositional inverse of formal power series to analytic issues is what allowed us in \cite{JKM1} to define a similarly analytically robust generalization of the Legendre transform by using algebraic Lagrange inversion. 
Concretely, we showed that the Legendre transform indeed possesses an underlying algebraic and combinatorial structure, namely as a map, $\fLT$, between formal power series.  $\fLT$ is locally finite and therefore it is robust in the sense that these power series no longer need to be convex nor do they even need to converge. Nevertheless, the algebraically-defined Legendre transform $\fLT$ obeys the usual transformation laws Equs.\eqref{e:aLT},\eqref{e:yLfg} for the Legendre transform.

Further, we showed that the algebraic and combinatorial structure of $\fLT$ possesses a beautiful interpretation that explains its robustness to analytic issues: the Legendre transform $\fLT$ can be shown to simply express the Euler characteristic of tree graphs, $V-E=1$, where $V$ is the number of vertices and $E$ is the number of edges. 

In the present paper we will be concerned mostly with the Fourier transform. Regarding the Legendre transform, we will prove one more property of the algebraically-defined Legendre transform $\fLT$, namely that the algebraically-defined Legendre transform $\fLT$ does not lose information even when mapping between power series of finite or zero radius of convergence. We will show that this is the case by proving that $\fLT$ is invertible, which we show by proving that $\fLT$ is a  quasi-involution (Theorem~\ref{T:QuasiInvolLT}).

\section{Combinatorial background}\label{S:EnBackG} 
In making the transition from analytic to  combinatorial arguments we now use the terminology of algebraic combinatorics. In particular,  "fields" are now  "formal power series" and a "generating functional" is a "generating series".  In this setting,
let $R$ be a ring, and $x$ is an indeterminate: that is it obeys only the absorption law: $x^m \cdot x^n = x^{m+n}$ where $m$ and $n$ are integers.   $R((x))$ is the ring of Laurent series in $x$ with a finite number of terms with negative exponents, and $R[[x]]$ is the ring of formal power series in $x$ over the coefficient ring~$R$. 
Also recall that any $r(x)\in R((x))$ has a canonical presentation $x^m \tilde{r}(x)$ where $\tilde{r}(x)\in R[[x]],$ and $\tilde{r}(0) \neq0$ ($m$ is the \emph{valuation} of $r(x)$). We note that $r(x)$ is invertible and that  $r^{-1}(x) = x^{-m} \,\tilde{r}^{-1}(x),$ where $\tilde{r}^{-1}(x)$ exists since $\tilde{r}(0) \neq0.$

Our approach, in general outline, is described with reference to Diagram~(\ref{CD:diagr2}). The top row is essentially~(\ref{e:ChainTranforms4c}), except the objects listed there are to be regarded in Diagram~(\ref{CD:diagr2}) as \emph{formal power series} rather than functions.  In the bottom row, $\cG^\ell$ is the set of all Feynman diagrams whose labelling scheme is to be determined,  $\cG^{\ell,c}$ is the subset of these that are connected, and $\cP_1^\ell$ is the subset of these that are 1PI-diagrams. The morphisms $\fLT^{\mathsf{comb}}$ and $\fFT^{\mathsf{comb}}$ are the combinatorial Legendre and Fourier transforms that are to be constructed, and the morphism $\BBbr$ constructs the generating series for the sets in the bottom row. 
\begin{equation}\label{CD:diagr2}
\begin{CD}
S @>\exp>>	e^S @>\fFT>>	Z @>\log>>	W @>\fLT>>	\Ga 					\\ 
@.			@.			@A{\BBbr}AA	@A{\BBbr}AA	@A{\BBbr}AA			\\ 
@.			@.			\cG^\ell @>>>	\cG^{\ell,c} 	@>\fLT^{\mathsf{comb}}>>	\cP^\ell_1	
\end{CD}
\end{equation}

A labelling scheme of certain constituents of Feynman diagrams is introduced to deal gracefully with automorphism groups.

The graphs in this Section have vertices and edges weighted by indeterminates, which later will be used to specify Feynman rules.

\subsection{Generating series}
A brief review of the necessary enumerative formalism is included. For more details see \cite{GJ}.

\subsubsection{Subobjects and associated weight functions}
Let $\cA$ be a set combinatorial objects or, more briefly, \emph{objects}, composed of generic $\sfs$-\emph{subobjects}  (\emph{e.g.} a graph composed of vertices and edges as its $\sv$- and $\sfe$-subobjects, respectively). The \emph{weight function} for $\sfs$-subobjects in $\cA$ is defined by
$$\om_{\sfs}\colon \cA \rar \bbN \colon a \mapsto \mbox{no. of $\sfs$-subobjects in $a$}.$$ 

There will be occasion to assign a distinct label to each of the $n$ $\sfs$-subobjects of $a$, thereby obtaining a \emph{labelled} object.  An $\sfs$-subobject that is not labelled is said to be of  \emph{ordinary type}, and one which is labelled is said to be of \emph{exponential type}.  Subobjects of the same type are said to be \emph{compatible}; otherwise, \emph{incompatible}.

Since it will be useful to have a highest label, we shall use  $\{1,\ldots,n\}$ with its natural linear order as the set of labels.  If $a$ also has $\sr$-subobjects that are labelled, then a disjoint set of labels $\{1', 2', \ldots\}$ will be used, where it is understood that the linear ordering is $1' < 2' < 3' <\ldots\;.$

\subsubsection{Generating series}\label{SSS:GenFns}
The ring of formal power series in an indeterminate $x$ with coefficients in a ring $R$ has the form
$
R[[x]] := \left\{ \sum_{k\ge0} a_k x^k
\; \colon \,\mbox{$a_k \in R$ for $k =0, 1,2, \ldots \; .$} \right\}.
$
To emphasise the distinction between the equality of functions and the equality of formal power series, we remark that two formal power series 
$ \sum_{k\ge0} a_k x^k \;\mbox{and}\;  \sum_{k\ge0} b_k x^k$
in $R[[x]]$ are equal if and only if $a_k = b_k$  for all $k\ge0.$ 

The ring $R((x))$ of formal Laurent series in an indeterminate $x$ with coefficients in a ring $R$ has the form
$
R((x)) := \left\{ \sum_{k\in \mathbb{Z}} a_k x^k
\; \colon \,\mbox{$a_k \in R$ for $k\in\mathbb{Z}$; $a_k\neq0$ for finitely many negative $k$ } \right\}.
$
If $f\in R((x))$ and $f = x^\alpha g$ where $g\in R[[x]]$ and $g(0)\neq0$ then $\alpha$ is said to be the \emph{valuation} of $f$; we write $\alpha = \val(f).$  It is noted that if $a,b\in R((x))$ then $\val(ab) = \val(a) + \val(b),$ a property shared by the degree of a polynomial.

The \emph{generating series} for the problem $(\cA,\wt{\sfs})$ in an indeterminate $x$ is
\begin{equation*}
\gens{\cA}{\sfs}(x) :=
\left\{
\begin{array}{cl}
\ds{\sum_{\si\in\cA} x^{\wt{\sfs}(\si)}} &\emph{if $\sfs$ is of ordinary type,} \\
\ds{\sum_{\si\in\cA} \frac{x^{\wt{\sfs}(\si)}}{{\wt{\sfs}(\si)}!}} &\emph{if $\sfs$ is of exponential type.}
\end{array}
\right.
\end{equation*}
The indeterminate $x$ is said to \emph{mark} the generic $\sfs$-subobject. For succinctness, we shall specify a subobject by writing, for example,
\begin{equation}\label{e:DefSubObj}
\sv:= \SOe{vertex}{x}
\end{equation}
to indicate that $\sv$ is a generic vertex associated with an exponential indeterminate~$x$.

\myss{Coefficient operator}
Let $R$ be a ring and $x$ an indeterminate.  The \emph{coefficient operator} $[x^m]$ is the linear operator defined by
$[x^m]\colon R[[x]] \rar R \colon \sum_{k\ge0} c_k x^k \mapsto c_m.$
Extension to $f\in R[[x_1, \ldots ,x_n]]$ is \textit{via} 
$
\left[x_1^{i_1} \cdots x_n^{i_n} \right] f = \left[x_1^{i_1} \cdots x_{n-1}^{i_{n-1}} \right]
\left( [x_n^{i_n} ] f \right)
$
under the natural ring isomorphism
$R[[x_1, \ldots ,x_n]] \cong R[[x_1, \ldots ,x_{n-1}]] \,[[x_n]].$

\myss{Differential operator} 
This is defined on $R[[x]]$ by
$
\frac{d}{dx} \colon R[[x]] \rar R[[x]] \colon x^k \mapsto k\, x^{k-1} 
$
for $k=0,1,2, \ldots,$ extended linearly to $R[[x]].$
It will be convenient on occasion to denote this operator by $\partial_x.$

\myss{Compositional inverse}
Let $a\in R[[x]]$. Then there exists $b\in R[[x]]$ such that $(a \oc b)(x) = x$  if and only if $[x^0]a = [x^0]b = 0$ and  $[x]a,  [x]b \neq 0$, in which case $(b \oc a)(x) = x$.  We denote the compositional inverse of $a$ by $\mi{a}$.   
 
\subsection{Enumerative lemmas}\label{SS:EnumLemmas}
\mbox{}\\
\myss{The set of canonical subsets} This set is defined to be
$\cU := \{ \varep, \{1\}, \, \{1,2\}, \, \{1,2,3\},\, \ldots\},$
where $\varep$ denotes the empty set. A generic element of any one of these sets is an $\sfs$-object
\begin{equation}\label{e:subobjGen}
\sfs = \SOe{element}{x},
\end{equation}
since the $\sfs$-objects are labelled.  Thus its generating series is
\begin{equation}\label{e:Uexp}
\gens{\cU}{\sfs}(x) = \sum_{k\ge0} 1\cdot \frac{x^k}{k!} = e^x.
\end{equation}

\myss{The set of canonical ordered sets} This set is defined to be $\cO :=\{\varep, \cO_1, \cO_2, \cO_3,\ldots\}$, where $\cO_n$ is the collection of ordered sets of the underlying set $\{1,2,\ldots,n\}$. A generic element of any one of these sets is an $\sfs$-object which is labelled. Thus its generating series is
\begin{equation}\label{e:Uexp}
\gens{\cO}{\sfs}(x) = \sum_{k\ge0} k!\cdot \frac{x^k}{k!} = \frac{1}{1-x}.
\end{equation}

\myss{Bijections}
Let $\cB$ be a set of combinatorial objects with $\sr$-subobjects. If $\Om\colon \cA\rar\cB$ is such that $\wt{\sfs} = \wt{\sr}\Om$, then $\Om$ is said to be an \emph{$\wt{\sfs}$-preserving} map. Clearly, if $\Om$ is also bijective, then
$\gens{\cA}{\sfs}(x) = \gens{\cB}{\sr}(x).$

\myss{Sum, product and composition lemmas}
Let $\cA_1$ and $\cA_2$ be disjoint subsets of $\cA.$ Then trivially, we have the \emph{Sum Lemma}
$\gens{\cA_1 \uplus \cA_2}{\sfs} = 
\gens{\cA_1}{\sfs} + \gens{\cA_2}{\sfs}.$

The product $\cA\boxt\cB$ of sets $\cA$ of combinatorial objects with $\sfs$-objects and
$\cB$ of combinatorial objects with $\sr$-objects, where $\sfs$ and $\sr$ are compatible is defined as the set of all $(a,b)\in \cA\times\cB$ in which all exponential $\sfs$- and $\sr$-subobjects are labelled in all possible ways.  If the subobjects are ordinary, then $\cA\boxt\cB$ is the Cartesian product of $\cA$ and $\cB.$ Let 
$\wt{\sfs}\oplus \wt{\sr} \colon \cA\boxt \cB \rar \bbN \colon
(a,b)\mapsto \wt{\sfs}(a) + \wt{\sr}(b).$
Then we have the \emph{Product Lemma}:
$\gensb{\cA\boxt\cB}{\wt{\sfs} \oplus \wt{\sr}} (x)
= \gens{\cA}{\sfs}(x) \cdot \gens{\cB}{\sr}(x). $
Implicit in this lemma is that the weight function for $\cA\boxt\cB$ is \emph{additive}.

The \emph{composition}  $\cA\oc_\sfs \cB$ is defined as the set of all objects constructed from $(a,b)\in\cA\times\cB$ by replacing each $\sfs$-subobject in $a$ in a unique way by each element $b$  of $\cB$ for all choices of $\fa$ and $\fb$.
The \emph{Composition Lemma} for   $\gens{\cA}{\sfs}(x)$ and $\gens{\cB}{\sr}(y)$ is:
$\gens{\cA\oc_\sfs\cB}{\sr}(y) =\left( \gens{\cA}{\sfs} \oc_x \gens{\cB}{\sr}\right)(y).$

\myss{The derivation lemma} There are two combinatorial operations, \textbf{deletion} and  \textbf{distinguishment} of an $\sfs$-subobject, for each type of sub-object.
\begin{enumerate}
\item [a)] \textbf{Deletion:}
Let $\ds{\ps\cA\equiv\frac{\partial}{\partial\sfs}{\cA}}$ be the set obtained from $\cA$ by:
\begin{enumerate} 
\item [i)] \textbf{if $\sfs$ is ordinary,}~ \emph{deletion} of a single $\sfs$-subobject of each $\si\in\cA$ in all possible ways, or 
\item [ii)] \textbf{if $\sfs$ is exponential,}~ \emph{deletion} of a canonical $\sfs$-subobject of each $\si\in\cA$, then 
\end{enumerate}
$$\gens{\ps\cA}{\sfs}(x)= \frac{d}{dx} \gens{\cA}{\sfs}(x).$$
\item [b)] \textbf{Distinguishment:}
Let $\ds{\sfs\ps\cA\equiv\frac{\sfs\partial}{\partial\sfs}{\cA}}$ be the set obtained from $\cA$ by \emph{distinguishment} of a single $\sfs$-subobject in each $\si\in\cA$ in all possible ways.  Then, for $\sfs$ \textbf{ordinary} or \textbf{exponential},
$$\gens{\sfs\ps\cA}{\sfs}(x)= x\frac{d}{dx} \gens{\cA}{\sfs}(x).$$
\end{enumerate}
We note that, trivially, $\sfs\ps \equiv \sfs\boxt\ps.$

\myss{Set theoretic coefficient extraction}
To complete the correspondences that we have drawn between purely combinatorial operations on sets of combinatorial objects, on the one hand and, on the other hand, algebraic operations on formal power series, it is convenient to include a combinatorial operation corresponding to the coefficient operator $[x^n]$.  That is, we introduce
\[
[\cN_n] \,\cA
\]
to denote the set of all elements of $\cA$ whose $n$ subobjects are labelled with $1, 2,\ldots,n$.

\myss{Multivariate generating series}
Let $\cA$ be a set of combinatorial objects with both $\sfs_1$- and $\sfs_2$-subobjects. Similarly, let $\cB$ be a set of combinatorial objects with both $\sr_1$- and $\sr_2$-subobjects.  We require that $\sfs_1$ and $\sr_1$ are compatible, and that $\sfs_2$ and $\sr_2$ are also. A weight function is defined on $\cA$ by
$\wt{\sfs_1} \ot \wt{\sfs_2} \colon\cA \rar \bbN^2 \colon \si 
\mapsto \left(\wt{\sfs_1}(\si),\wt{\sfs_2}(\si)\right),$
and similarly for $\cB$.
It is convenient to regard $\wt{\sfs_1} \ot \wt{\sfs_2}$ as a \emph{refinement} of the weight  $\wt{\sfs_1}$ by  the weight $\wt{\sfs_2}.$  Then
$
\left(\wt{\sfs_1} \ot \wt{\sfs_2}\right) \oplus \left(\wt{\sr_1} \ot \wt{\sr_2}\right)
= \left(\wt{\sfs_1} \oplus \wt{\fr_1}\right) \ot \left(\wt{\sfs_2} \oplus\wt{\sf_2}\right).
$
The Product Lemma holds under this refinement of the weight function. 

\myss{Convention for exponential and ordinary arguments of generating series} 
For clarity and succinctness, we shall adopt the following convention for exponential and ordinary indeterminates.

By convention, arguments of  $f(x_1,\ldots; y_1, \ldots)$ before the  semi-colon are exponential and the remaining ordinary.  The weight functions similarly ordered. We shall indicate the absence of exponential and ordinary indeterminates, respectively, in $f$  by $f( - ; y_1, \ldots)$ and $f(x_1, \ldots ; - ).$ 
The presence of a semi-colon in the argument list indicates that this convention is being invoked.  If there is no semi-colon, then no such assertion is being made.

\myss{The Inclusion-Exclusion Lemma}
Let $\cA$ be a set of combinatorial objects with $\sr$-~and $\sfs$-subobjects.  Let $\cA^\star$ be the multiset of all objects constructed from each $a\in\cA$ by distinguishing a subset of at least $k$ $\sr$-subobjects in all possible ways, and for all values of $k$.  Then 
$\gensb{\cA^\star}{\om_\sfs\ot\om_\sr}(x,y) = \gensb{\cA}{\om_\sfs\ot\om_\sr}(x,y+1).$ 
In particular, the generating series for the elements of $\cA$ with no $\sr$-subobjects is
$\gensb{\cA}{\om_\sfs\ot\om_\sr}(x,0) = \gensb{\cA^\star}{\om_\sfs\ot\om_\sr}(x,-1).$
It is noted that $\cA^\star$ is defined by weaker conditions than $\cA$.

\myss{Sign-reversing involution Lemma}
Let $\cS$ be a set of combinatorial objects with $\sfs$-subobjects and let $w$ be a function on $\cS.$ Let $\cS = \cA\uplus\cB$ be a partitioning of $\cS$ such that there exist a map $\Om\colon\cB\rar\cB$ and a function $\widehat{w}$ on $\cS$ with the properties:
\begin{align*}
(a)\quad& \Om^2 = \sI,	&\explain{$\Om$ is involutory},	\\
(b)\quad& \widehat{w}\left(\Om(\si)\right) = -\widehat{w}(\si) \mbox{  for all $\si\in\cB$}, &\explain{$\Om$ is sign-reversing},	\\
(c)\quad& \widehat{w}(\si) = w(\si) \mbox{  for all $\si\in\cA$},
&\explain{$\widehat{w}$ properly restricts to $\cA$}.	
\end{align*}
Then
$\sum_{\si\in\cA} w(\si) =  \sum_{\si\in\cS} \widehat{w}(\si).$

\subsection{Decomposition of graphs}\label{SSS:DecompGraph}
We shall need the following subobjects, and the operation between them, in decomposing graphical objects.  Loops and multiple edges are allowed in graphs unless otherwise stated. It is to be assumed that vertices are \emph{unlabelled} unless otherwise stated.  

\myss{Pre-graphs and the glueing operation $\Join$}
A \emph{pre-edge} is an edge with \emph{open-ends}, and is represented combinatorially by 
$\myputeps{1}{0.45}{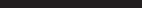}$.  A \emph{pre-vertex} is an element from the list 
$$ \myputeps{-0.2}{0.40}{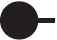},\, \quad\myputeps{-0.2}{0.40}{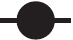},\, \quad\myputeps{-1.00}{0.40}{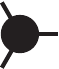},\, \quad\myputeps{-1.10}{0.40}{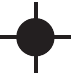},\,  \ldots\: , $$
of vertices with short line segments called \emph{attachment points}. The \emph{glueing} operation $\Join$ is a binary operation that takes two arguments, a point of attachment and a open-end of a pre-edge, and identifies, or glues, one with the other. A \emph{pre-graph} is obtained by a finite number of applications of the binary operation $\Join$ of \emph{glueing} an open-end of a pre-edge to an attachment point.
For example,  
\begin{equation*}\label{e:GlueOp1}
\myputeps{-1.2}{0.40}{3Vertex} \!\! \bowtie 
\myputeps{0.75}{0.45}{preedge}
:= 
\myputeps{-1.4}{0.40}{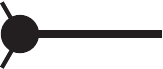}.
\end{equation*}

\myss{The anti-edge}
The glueing operation is extended to include a fictitious element   
$\Bmyputeps{0.5}{0.45}{preedge}^{-1}$,  called an \emph{anti-edge}. It has the single property 
\begin{equation*}\label{e:AntiEdge1}
\Bmyputeps{0.5}{0.45}{preedge}\bowtie \Bmyputeps{0.4}{0.45}{preedge}^{-1}
=\Bmyputeps{0.35}{0.45}{preedge}^{-1}\bowtie \Bmyputeps{0.5}{0.45}{preedge}
= \mathrm{null}.
\end{equation*}

\begin{remark}\label{R:AntiEdge}
The anti-edge will be used as a combinatorial device for preserving the number $\om_{\sfe}$ of edges in decompositions of graphs.
\end{remark}

\myss{Two types of $1$-vertex}
We shall have occasion to use two types of $1$-vertices:
\begin{itemize}
\item $1$-vertices that are labelled (and are marked by an exponential indeterminate $y$);
\item $1$-vertices that are not labelled  (and are marked by an ordinary indeterminate $\lam_1$).
\end{itemize}

\subsection{General results for $R[[x]]$}

Several useful results are gathered here for power series. The first few involve the exponential series in an essential way in the context of differential operators. The last two results are the well known residue composition and Lagrange inversion for formal power series.



\begin{lemma}\label{L:EigVal}
Let $x, y$ and $\al$ be indeterminates $a \in R[[x]]$ and let $\al\in R.$   Then
$a(t\pd{x}) e^{\al xy} = a(\al t y) e^{\al xy}.$
\end{lemma}

\begin{lemma}[Taylor's theorem]\label{L:TaylThm}
Let  $h\in R[[x]]$ and let $t$ be an indeterminate. Then
$e^{t\partial_x} h(x) = h(x+t).$
\end{lemma}

\begin{lemma}\label{L:ArgShift}
Let  $h\in R[[x]]$.  Then
$
h(x) = [z^0]\, h(\partial_z)e^{xz}.
$
\end{lemma}
Lemma~\ref{L:TaylThm} is the ring theoretic counterpart of Taylor's Theorem, and  Lemma~\ref{L:ArgShift} may be regarded as a device for moving the argument $x$ of $h$ into the exponent position, which then interdigitates well with Taylor's Theorem.

\begin{lemma}
Let $x$ and $t$ be indeterminates. Then
$
e^{t\pd{x}} \colon R[[x]] \rightarrow R[[x,t]]
$
is a ring homomorphism.
\end{lemma}
\begin{proof}
Let $A, B\in R[[x]].$  Trivially, in the case of addition, we have
$e^{t\pd{x}}  (\lam A + \mu B) =  \lam e^{t\pd{x}} A +  \mu e^{t\pd{x}} B$
where $\lam, \mu \in R.$
In the case of product, let $n$ be a non-negative integer.  Then, by Lemma \ref{L:TaylThm},
$e^{t\pd{x}} (A\cdot B) = A(x+t)\cdot B(x+t) = \left(e^{t\pd{x}} A\right)\cdot \left(e^{t\pd{x}} B\right). $
The result now follows.
\end{proof}

\begin{lemma}\label{TestZero}
Let $R_1$ and $R_2$ be rings, and let
$
T\colon R_1[[x]] \rightarrow R_2[[x]]
$
be a linear transformation. Let $t$ be an indeterminate. Then
$
T e^{tx} = 0 \quad\mbox{if and only if} \quad T f = 0 \quad \mbox{for all $f\in R_1[[x]]$.}
$
\end{lemma}
\begin{proof}
If $T e^{tx} = 0$ then $[t^n] T e^{tx} = 0$ for all $n\ge0$ so $T x^n=0$ for all $n\ge0$, and so $T f = 0$  for all $f\in R_1[[x]]$. This argument is reversible, and the result follows.  
\end{proof}

\begin{proposition}[{Exponential formula \cite[Prop. 5.1.7]{EC2}}]
Let $A(x) = \sum_{n\geq 1} a_n x^n/n!$ and $B(x) = 1+ \sum_{n\geq 1} b_n x^n/n!$ be power series in $R[[x]]$. If $B(x) = e^{A(x)}$ then for $n\geq 0$ we have
\[
a_{n+1} = b_{n+1} - \sum_{k=1}^n \binom{n}{k} b_k a_{n+1-k}.
\]

\end{proposition}

For proofs of the next theorem, Theorem~\ref{T:ResComp}, see for example, Section 1.2.2 \cite{GJ}. 

\begin{theorem}[Residue Composition]\label{T:ResComp}
Let $f(x), r(x)\in R((x))$ and let $\val(r) = \alpha >0.$ Then
$
\alpha \, [x^{-1}] f(x) = [z^{-1}] \, f\left(r(z)\right) \, r'(z).
$
\end{theorem}

\begin{theorem}[Lagrange Inversion Theorem]\label{T:LIT}
Let $\phi(\lam) \in R[[\lam]]$ with $\phi(0)\neq 0$. Then the functional equation
$ 
w=t \cdot \phi(w)
$
has a unique solution $w(t) \in R[[t]].$ Moreover, if $f(\lam) \in R[[\lam]]$ then 
$$[t^k] f(w) = \frac{1}{k}[\lam^{k-1}] f'(\lam)\phi^k(\lam) \quad\text{ for } k>0
\qquad\mbox{and}\qquad [t^0]f(w)=f(0).$$
\end{theorem}

The next subsection has a self contained combinatorial proof of this theorem. An immediate corollary is an expression for the compositional inverse 
$\mi{a}(x)$ when it exists.
\begin{corollary}\label{C:CompInv}
If $a\in R[[x]]$ has a compositional inverse, then it is given by
$\mi{a}(x) = \sum_{k\ge1}  \frac1k x^k [t^{k-1}] \phi^{-k}(t)$
where $\phi(t)$ is such that $a(x) = x\cdot\phi(x)$ and $\phi(0)\neq0.$
\end{corollary}
\begin{proof}
Let  $t:= a(x)$ so $x=\mi{a}(t)$ and $x$ satisfies the functional equation $x=t \phi^{-1}(x).$ The result follows from Theorem~\ref{T:LIT}.
\end{proof}

\subsection{Combinatorial proof of Lagrange inversion theorem}

The purpose of this section is to realise Lagrange's Inversion Theorem in a combinatorial form.  This will be done later for the Legendre and Fourier transforms. For other algebraic and combinatorial proofs of Lagrange's inversion theorem see the survey by Gessel \cite{Ge} and \cite[Sec. 1.2.4]{GJ}.

\textbf{Background: }
Let $\cT$ be the set of all rooted labelled trees. On deletion of its root vertex, a tree $t\in\cT$ decomposes into an unordered set of trees in $\cT$, so
\begin{equation}\label{e:new1}
\cT \stackrel  {\sim} {\longrightarrow} \{\bullet\} \cmul (\cU \circ \cT)
\end{equation}
where `$\bullet$' denotes the root vertex. The exponential generating series for $\cT$ with respect to the number of vertices therefore satisfies the functional equation
\begin{equation}\label{e:new2}
T(x) = x P(T(x)) \quad\mbox{where}\quad P(x) := e^x.
\end{equation}
Its unique solution in $\sR[[x]]$ is given by
$\left [ \frac{x^n}{n!} \right ] T(x)  = n^{n-1},$ 
suggesting the existence of a natural combinatorially defined bijection
\begin{equation}\label{e:new3}
\om_n \colon \left  [ \cN_n\right ] \cT \stackrel  {\sim} {\longrightarrow} 
{\cN_n}^{\cN_{n-1}},
\end{equation}
This \textit{constructs an explicit solution} of~(\ref{e:new2}). 
We seek a refinement $\Omega$ of (\ref{e:new3}) that \emph{constructs an explicit solution} for arbitrary $P(x)$.  This, perforce, will be a combinatorial form of Lagrange's Inversion Theorem. 

\textbf{The refinement}: 
Let $\cA$ be the set of all labelled rooted trees with a root-vertex marked by $0$ (this marking remains unchanged throughout the constructions), and let $\cA_k$ be the subset of  $\cA$ with root-degree $k$.  Let $[\cN_n]\cA_k$, where $\cN_n := \{1, 2,  \ldots, n\}$, denote the set of all trees in $\cA_k$ with non-root vertices labelled from $1$  to $n$.  It will suffice to consider an example with $t\in\left[\cN_{15}\right] \cA_4$ since the general case  will follow immediately. We now describe a sequence of four transformations applied to $t$. 
 
\textbf{(i)} For a tree $t\in \cA_k$, let $t \mapsto t'$ where $t'$ is obtained by directing each edge of $t$ towards the root-vertex of $t$. 
\[
t \longmapsto t' := \quad \raisebox{-4mm}{\includegraphics[scale=0.9]{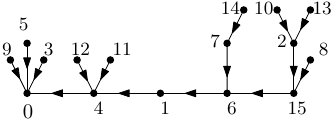}}
\]
The  \emph{in-degree}, $d^+_{t'}(i)$, of a vertex labelled $i$ in $t'$ is the number of edges of $t'$ entering vertex $i.$ Thus, $d^+_i(t')= k$. The intention is to preserve this weight function throughout the transformations.

\textbf{(ii)}  The unique path $4 \leftarrow 1 \leftarrow 6$ obtained by deleting the terminal vertices $0$ and $15$ of the path between $0$ and $n=15$ in $t$ may be encoded as the permutation
$
\left(
\sze{$\begin{array}{ccc}
 1 & 4 & 6\\ 4 & 1 &6 
\end{array}$
}
\right)
$.
The first row lists the elements of the second row in increasing order. The disjoint cycle representation of this permutation is
$
\left(
\sze{$\begin{array}{cc}
1 & 4\\ 4 & 1 
\end{array}
$}
\right)
\sze{$\left(\begin{array}{c}
6\\ 6  
\end{array}
\right).$}
$ 
Reattaching these trees by their roots vertices to the two cycles gives 
$$
t' \longmapsto g:= \quad \raisebox{-4mm}{\includegraphics[scale=0.9]{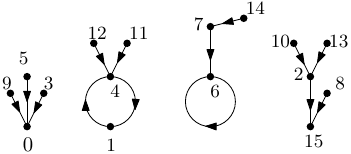}}
$$

\textbf{(iii)} Now let 
$g\longmapsto f \mbox{  where  } f: \cN_{n-1} \longrightarrow \{0,1, 2, \ldots ,n\}: p \longmapsto q$
if there is an edge in $g$ from $p$ to $q$. We note, \emph{obiter dictu}, that $g$ is called a functional digraph.

\textbf{(iv)} Finally, let
$
f \longmapsto \cB := (\cB_0, \cB_1, \ldots, \cB_n) \in \left[\cN_{n-1}\right] \left(\cN_{k-1} \star (\cO_n\circ \cU)\right)
$
where $\cB_i := f^{-1}(i)$ is the pre-image set of $i$ under the action of $f$, for $i=0, \ldots, n$. Note that  $\cB$ is an ordered  set-partition of $\{0, 1, \ldots,n\}$ where empty parts are allowed.   For instance, in the current example,
$\cB_0 =\{3, 5, 9\}, \cB_1 =\{4\}, \cB_2 =\{10, 13\}, \cB_4 =\{1, 11, 12\}, \cB_6 =\{7, 15\}, \cB_7 =\{14\}, \cB_{15} =\{2, 8\}, $ and $\cB_i = \emptyset$ for $i\in \{3, 5, 7, 8, 9, 10, 11, 12 ,13, 14\}$. 

The composition $t \mapsto t' \mapsto g \mapsto f \mapsto \cB$  of these actions gives a bijective map 
\begin{equation}\label{new6}
\left[\cN_n\right] \cA_k \stackrel  {\sim} {\longrightarrow}  
\left[\cN_{n-1}\right] \left(\cN_{k-1} \star (\cO_n\circ \cU)\right)
\colon t \longmapsto (\cB_0, \cB_1, \ldots, \cB_n),
\end{equation}
 with the in-degree preserving property that
\begin{equation}\label{new6a}
d^+_i(t) =\left\{
\begin{array}{cl}
|\cB_i | & \mbox{for  }  i=1,\ldots, n,\\
1+ |\cB_0 |  = k& \mbox{otherwise.}
\end{array}
\right.
\end{equation}
But $\cN_{k-1} = \frac{\partial}{\partial \sfe} \cN_{k}$ by deleting the highest labelled element ($\sfe$-subobject) of $\cN_k$, and  $\cA_k \stackrel  {\sim} {\longrightarrow}  \cT^k \stackrel  {\sim} {\longrightarrow}  \cN_k \circ \cT$.  
Then from~(\ref{new6}), we have the bijective map
\begin{equation}\label{new7}
\Omega :=  \left[\cN_n\right] (\cN_k\circ\cT) 
\stackrel  {\sim} {\longrightarrow}
 \left[\cN_{n-1}\right]
\left(\frac{\partial\cN_{k}}{\partial \sfe}  \star (\cO_n\circ\cU)\right)
\colon t \longmapsto (\cB_0, \cB_1, \ldots , \cB_n).
\end{equation}
In summary, we have a bijection between and unordered $k$-set of labelled rooted trees (also called functional digraphs, as observed above) on $\cN_n$ and their collection of pre-image sets.  This bijection preserves the in-degree of each vertex vertices $i$  in the cardinality of $\cB_i$.

\textbf{Connexion with Lagrange's Inversion Theorem:}  
To justify referring to~(\ref{new7}) as a combinatorial form of Lagrange's Inversion Theorem, we use it to prove the algebraic form of the theorem.  Let $b_k$ mark a root vertex of in-degree k.  Then multiplying both sides of (\ref{new7}) by $b_k$, taking the disjoint union over $k$ and constructing their generating series, we have
$$
\left[\frac{x^n}{n!}\right] \sum_{k\ge1} \frac{1}{k!} b_k T^k(x) 
= \left[\frac{x^{n-1}}{(n-1)!}\right]  \left( \frac{d}{dx}\sum_{k\ge1} b_k \frac{x^k}{k!}\right) \cdot 
\llbracket  \cO_n\circ\cU\rrbracket_e(x).
$$ 
Then, replacing the indeterminate $x$ on the right hand side by an indeterminate $\lam$,  
\begin{equation}\label{new7a}
\left[x^n\right] \Phi(T(x)) = \frac{1}{n}\left[\lam^{n-1}\right] \Phi'(\lam) \llbracket \cO_n\circ\cU\rrbracket_e(\lam) 
 \quad\mbox{where $\Phi(\lam) := \sum_{i\ge0} b_i \frac{\lam^i}{i!}$}.
\end{equation}
Let $P(\lam) :=  \llbracket \cU\rrbracket_e(\lam) = \sum_{i\ge0} a_i \frac{\lam^i}{i!}$, where $a_i$ marks a pre-image set of size $i$. Then
$ \llbracket \cO_n\circ\cU\rrbracket_e(\lam)  = P^n(\lam).$ 
But, from~(\ref{new6a}),  $a_i$ may be replaced by $b_i$ that marks a vertex in $t\in\cT$ with in-degree $i$.  
Then, from~(\ref{e:new1}),  
\begin{equation}\label{new9}
T(\lam) = \lam  P(T(\lam))
\end{equation}
and, from (\ref{new7a}),
\begin{equation}\label{new8}
\left[x^n\right] \Phi(T(x)) = \frac{1}{n} \left[\lam^{n-1}\right]  \Phi'(\lam) P^n(\lam).
\end{equation}
 
Setting aside the combinatorial model behind assertions~(\ref{new9}) and ~(\ref{new8}), we may restate them formally as follows.
If $T(\lam)$ is a formal power series that satisfies the functional equation~(\ref{new9}), and $\Phi(\lam)$ is a formal power series, then
$\Phi(T(x))$ is given explicitly by~(\ref{new8}).  This is precisely the \emph{algebraic form of Lagrange's Inversion Theorem}. This having been established, we may now assert that~(\ref{new7}) is the \emph{combinatorial form of Lagrange's Inversion Theorem}.

\section{The formal and diagrammatic Fourier transform: univariate case}\label{S:FormalFT1}

\newcommand{\mydel}[1]{}

\newcommand{\pa}{\partial_{a}}
\newcommand*{\I}{\imath}
\newcommand{\mrd}{\mathrm{d}}
\newcommand{\mP}{p}
\newcommand{\mQ}{q}
\newcommand{\pw}{\partial_{w}}
\newcommand{\px}{\partial_{x}}
\newcommand{\py}{\partial_{y}}
\newcommand{\pz}{\partial_{z}}

\newcommand{\PQ}{\Delta}
\newcommand{\sfO}{\mathsf{O}}
\newcommand{\sfQ}{\mathsf{Q}}

\newcommand{\sha}{\tfrac{1}{2}}	
\newcommand{\fidy}{f(i \py)}

We consider a combinatorial bijection that acts on connected labelled graphs.  In view of its r\^{o}le on Feynman diagrams, we propose to refer to it as a \emph{combinatorial Fourier transform}. To lend further support to this terminology, it will be shown that it possesses properties analogous to those of the analytic Fourier transform $(\FT f) (y) := \frac{1}{\sqrt{2\pi}} \int_\bbR f(x) e^{-i\,xy} dx\quad\mbox{for every $y\in\bbR$}$ and that it agrees with $\FT$ on the physically particularly important Hilbert space $L_2(\mathbb{R})$ of square integrable functions. In fact we will prove more, namely the  analogues of:

\begin{property}\label{PRPTY:Fourier} 
\begin{align*}
(a)\quad&\left(\sfF f(x+c)\right)(y) = e^{-icy} \left(\sfF f(x)\right)(y),
&\explain{\rm translation invariance: where $c$ is a scalar}	\\
(b)\quad&\left(\sfF (h(i\pd{x}) f)\right)(y) = h(y) \cdot \left(\sfF f\right)(y),
&\explain{\rm product-derivation I: where $h\in R[[x]]$}	\\
(c)\quad &
(\sfF (h\cdot f))(y) = h(-i\pd{y})\cdot (\sfF f)(y), 
&\explain{\rm product-derivation II: where $h\in R[[x]]$}\\
(d)\quad&\left(\sfF\left(\sfF f\right)\right)(x) = f(-x),
&\explain{\rm quasi-involutory}\\
(e)\quad &  (\sfF (e^{-x^2/2}))(y) = e^{-y^2/2}.  &\explain{\rm value on Gaussians}
\end{align*}
\end{property}

Combinatorial proofs will be given for the analogues of the last three properties in Section~\ref{SS:CombAppFourier}. Algebraic proofs of the analogues of all of these properties for our transform denoted by $\fFT_a$ will be given in  Section~\ref{SS:AlgAppFourier}. At this point we can see that proving the analogues of properties (c) and (e) implies that our transform $\fFT_a$ agrees with the analytic Fourier transform $\sfF$ on the Hilbert space of square integrable functions $L_2(\mathbb{R})$. This is because our transform $\fFT_a$ satisfies the analogues of product-derivation property and has the same value on Gaussians as $\sfF$ then it agrees on the product of polynomials and Gaussians. These functions include \emph{Hermite functions} which form an orthonormal basis of the Hilbert space of square integrable functions $L_2(\mathbb{R})$. For the details see Corollary~\ref{cor: f equals F_a}.

\subsection{The basic combinatorial structure, and its labelling convention}\label{SS:FormalFT1}
The Subdiagram
\begin{equation}\label{CD:diag12}
\begin{CD}
S @>\exp>>	e^S @>\FT>>	Z@>\log>>	W 			\\  
@.			@.			@A{\BBbr}AA	@A{\BBbr}AA	\\ 
@.			@.			\cG^\ell @>>>	\cG^{\ell,c} 		
\end{CD}
\end{equation}
of Diagram~(\ref{CD:diagr2}) suggests that the combinatorial structure and its labelling convention will be decided by the action of the Fourier transform $\FT$ on $e^S$, and Diagram~(\ref{CD:diagr2}) suggests that the action of $\LT$  is to be compatible with this structure.

To determine the structure, it is sufficient to consider the univariate case of 0-dimensional QFT (i.e. QFT with one field on a $0$-dimensional space). Briefly, for the analytic Fourier transform, this looks as follows:

\subsubsection{Preliminaries}\label{SSS:prelims:AlgFourTrans}
Let
$f(x) =\frac{1}{2} u^{-1}x^2 - g(x)$ where $ g(x) = \sum_{k\ge2} \lam_k \frac{x^k}{k!}.$ 
Then
\begin{equation*}\label{e:Xfy}
\left(\FT e^{i\,f} \right)(y) =  \frac{1}{\sqrt{2\pi}}  \int_\bbR e^{-i\,g(x)} \cdot e^{i\, (\frac{1}{2} u^{-1} x^2 - xy)}dx.
\end{equation*}
But, for any function $a,$ we have $a(\partial/\partial y) e^{xy} = a(x) e^{xy}.$ Then
$$
\left(\FT e^{i\,f}\right) (y)
= \frac{1}{\sqrt{2\pi}}  \int_\bbR e^{-i g(i\,\partial/\partial y)} \cdot  e^{i\, (\frac{1}{2} u^{-1} x^2 - xy)}dx 
=  \frac{1}{\sqrt{2\pi}} e^{-i g(i\,\partial/\partial y)} \: \int_\bbR e^{i\, (\frac{1}{2} u^{-1} x^2 - xy)}dx.
$$
Changing variables through $\xi := x - uy$, we have  $\frac{1}{2} u^{-1} x^2 - xy = \frac{1}{2} u^{-1}\xi^2 -\frac{1}{2} uy^2$ so
\begin{equation*}\label{e:PrelimFtrans}
\left(\FT e^{i\,f} \right)(y) = C\cdot e^{-i g(i\,\partial/\partial y)} \: e^{-\frac{1}{2}i\, uy^2}  \qquad\mbox{where}\qquad 
C :=  \sqrt{iu}  \int_\bbR e^{-\frac{1}{2i} u^{-1} \xi^2} d\xi = \sqrt{2\pi\, i u}.
\end{equation*}
For the latter integral to exist, $u$ is required to be real and positive. 
Attention is therefore drawn to the expression 
\begin{equation}\label{e:expGDiffE}
\left. e^{g(x)}\right|_{x\mapsto \partial_y}  e^{\frac 12  u y^2}
\end{equation}
that suggests the statement for an algebraic Fourier transform.

\begin{definition}[algebraic Fourier transform] \label{D:fFTa}
Let $f$ be a formal power series in an indeterminate $x,$ such that $u := [x^2]f(x) \neq 0.$
Let $a \neq 0$ be an indeterminate. The transform $\fFT_a$ is defined by its element action
\begin{align*}
\mathrm{(a)}\quad \fFT_a[e^f](y)  &\coloneqq \frac{1}{\sqrt{a}}
\cdot e^{f(i\, \py)} \, e^{-\ha a \py^2}  \, e^{-\ha a^{-1} y^2}
\intertext{or, equivalently,}
\mathrm{(b)}\quad \fFT_a[e^f](y)  &\coloneqq \frac{1}{\sqrt{a}}
\cdot \left. e^{f(x)} \, e^{\ha a x^2}\right|_{x\mapsto i\,\py}  \, e^{-\ha a^{-1} y^2}.
\end{align*}
\end{definition}

\begin{theorem}[Main Theorem for the algebraic Fourier transform]
\label{thm:fFTa}
Let $f$ be a formal power series in an indeterminate $x,$ such that $u := [x^2]f(x) \neq 0.$
Let $a \neq 0$ be an indeterminate. Then $\fFT_a$ is a formal power series in the indeterminates $y$ and $u$,  independent of $a$, and it is an algebraic Fourier transform, i.e. it satisfies analogues of properties listed in Property~\ref{PRPTY:Fourier}.
\end{theorem}

\begin{remark} \label{rem:algtranslocallyfinite}
Note that since $\fFT_a[e^f]$ is in a ring of formal power series with indeterminates $y$ and $u$, its coefficients in these indeterminates are locally finite. 
\end{remark}

\begin{remark}

While our theorem holds for all $a\neq 0$, in the special case $a>0$, the fact that the expression $\left.e^{\ha a x^2}\right|_{x\mapsto i\,\py}  \, e^{-\ha a^{-1} y^2}$
is independent of $a$ can also be seen analytically. This is because in this case, $e^{-\ha a^{-1} y^2}$ represents a Gaussian function and the evolution operator $\left.e^{\ha a x^2}\right|_{x\mapsto i\,\py}$ evolves this Gaussian backwards in time according to the heat equation so that it becomes a Dirac delta, $\delta(y)$. This results is independent of $a$ because no matter how large we choose $\vert a\vert$, i.e., no matter how wide we choose the Gaussian, the evolution operator depends on $a$ such that it evolves the Gaussian back into the Dirac delta. In this context, see also our related results on integration by differentiation in \cite{JKM2,JKM3}. 
\end{remark}

The proof of this Theorem is deferred to Section~\ref{SSS:MainThmPf} since preliminary results are required.

\begin{corollary} \label{cor: f equals F_a}
If $h(x)$ is a square integrable function in $L_2(\mathbb{R})$ then $(\sfF \, h)(y)= (\fFT_a \,h)(y)$.
\end{corollary}

\begin{proof}

We start by observing that the space of square integrable functions are spanned by \emph{Hermite functions} $\Omega_n(x):=(1/(\pi^{1/4}2^{n/2}\sqrt{n!}))H_n(x) e^{-x^2/2}$, where $H_n(x)$ is a \emph{Hermite polynomial} (e.g. see \cite[Ch. 6]{AAR}). By Property~\ref{PRPTY:Fourier}(b)(d) of the analytic Fourier transform we have that
\[
(\sfF\, \Omega_n(x))(y) = \frac{1}{\pi^{1/4}2^{n/2}\sqrt{n!}} H_n(-i\partial_y) \cdot e^{-y^2/2}
\]
On the other hand, by Lemma~\ref{ProdDerivII} and Lemma~\ref{lem: value Gaussians} we have that 
\[
(\fFT_a\, \Omega_n(x))(y) = \frac{1}{\pi^{1/4}2^{n/2}\sqrt{n!}} H_n(-i\partial_y) \cdot e^{-y^2/2}
\]
Since both $\sfF$ and $\fFT_a$ agree on a basis of the space of square integrable functions then the result follows by linearity of the transforms.
\end{proof}

\subsection{Proof of the Theorem~\ref{thm:fFTa}} ~\label{SS:AlgAppFourier}
To prove that $\fFT_a$ from Definition~\ref{D:fFTa} is indeed the formal Fourier transform, it is necessary to show that it has the properties listed in Properties~\ref{PRPTY:Fourier}.  In the algebraic case, there is an additional property that  it acts between formal power series (or Laurent series) rings.  These four properties are proved in the following four lemmas.

\subsubsection{An independence result}
Let $f(x)$ be a formal power series in $x$ whose coefficients do not depend on  $a$.
We determine a sufficient condition on $\mP(a), \mQ(a)$ and $\theta(a)$  
for 
\begin{equation}\label{e:Phi1}
\Phi_{p,q}(y)   \coloneqq \,\theta(a)  \, e^{\fidy} \,e^{\ha\mP(a) \py^2} \, e^{\ha\mQ(a) y^2}
\end{equation}
to be independent of $a$: that is, that $\pa\Phi_{p,q} = \sfO(y)$, where $\sfO(y)$ is the zero series. To this end, let 
$$G\coloneqq e^{\fidy} e^{\ha\mP(a) \py^2} \, e^{\ha\mQ(a) y^2}
\quad\mbox{and}\quad 
B \coloneqq e^{\fidy} e^{\ha \mP(a) \py^2}\, y^2 \, e^{\ha \mQ(a) y^2}.
$$ 
Then 
$$\Phi_{p,q}(y)  = \theta(a) \cdot G,$$
and it  is a straightforward matter to show that
\begin{equation*}
\pa G = \sha \mP'(a) \,\py^2 G +\sha \mQ'(a) \, B, \quad\mbox{and}\quad
\py^2 G = \mQ(a) G + \mQ^2(a) B, 
\end{equation*}
from which, on elimination of $B$,
\begin{equation}\label{ee:PQ10}
\mQ^2(a) \, \pa G + \sha \mQ'(a)\, \mQ(a) \,G%
=  \sha \left( \mP'(a) \,  \mQ^2(a) + \mQ'(a)\right) \py^2 G.
\end{equation}
Now assume that $q(a) \neq0.$ Then the above expression may be rewritten in the form
\begin{equation}\label{e:newequ1}
\pa\left(\mQ^{\ha}(a) \cdot G\right)
= \sha\, \mQ^{\ha}(a) \cdot \left(\pa(\mP'(a) - \mQ^{-1}(a)) \right)\cdot \py^2 G,
\end{equation}
provided $q(a)$ is invertible (that is, it is not the zero series).
Set $\theta(a) =  q^{\ha}(a)$ and $\pa\left((\mP'(a) - \mQ^{-1}(a)\right) = 0$, so 
$$\mP(a) = \mQ^{-1}(a) - \be \quad\mbox{where}\quad \be\in\bbC.$$  Then the left  and right hand sides of~(\ref{e:newequ1}) is equal to $\pa\Phi_{p,q}(y) $ and ~$\sfO(y),$ respectively, so 
$$\pa\Phi_{p,q}(y)  = \sfO(y).$$
Now let 
$$f(x) = \lam_2 \frac{x^2}{2!} + g(x), 
\quad\mbox{where}\quad 
g(x) = \sum_{k\ge 1,\, k\neq2}\lam_k \frac{x^k}{k!} \quad\mbox{and}\quad \lam_2\neq0,
$$
 set $\be =  \lam_2$. 
 By Lemma~\ref{L:PQexp},
$$
1 \, \sqrt{\mQ(a)}\cdot  e^{g(i \py)} \, e^{\ha (\mQ^{-1}(a) - \lam_2) \, \py^2} \, e^{\ha \mQ(a)y^2}.
$$
is a formal power series in $y$ provided 
$(\mQ^{-1}(a) - \lam_2) \cdot \mQ(a) \neq 1.$ That is, $\lam_2\cdot  \mQ(a) \neq0.$ 
But this is so since $q(a) \neq0$ since it is invertible, and since $\lam_2\neq0.$

This proves the following result.

\begin{lemma}\label{L:pq:beta}
Let   $\mQ(a)$, different from the zero series $\sfO(x)$, be a formal series in an indeterminate $a$.  Let
$f(x)$ be a formal power series in $x$ whose coefficients are independent of $a$, and such that $[x^2]f\neq0.$

Then
$$
\sqrt{\mQ(a)}\cdot  e^{f(i \py)} \, e^{\ha \mQ^{-1}(a) \, \py^2} \, e^{\ha \mQ(a)y^2}
$$
is   
(a)  a formal power series in $y$ and (b) independent of $a$.
\end{lemma}

\subsubsection{The independence-series property}
\begin{lemma}[Independence-series]\label{L:IndSer}
Let $f(x)$ be a formal power series in $x$ with coefficients independent of the indeterminate $a$, and such that $[x^2]f(x)\neq0.$  Then
$\fFT_a[e^f](y)$ is a formal power series in $y$, and is independent of $a$.
\end{lemma}
\begin{proof}
Setting $q(a) = -a^{-1}$ in Lemma~\ref{L:pq:beta}, it follows that
$$\frac{1}{\sqrt{a}}
\cdot e^{f(i\, \py)} \, e^{-\ha a \py^2}  \, e^{-\ha a^{-1} y^2}$$
is a formal power series  in $y$ and is independent of $a$. But this is the definition of $\fFT_a[e^f]$ in Definition~\ref{D:fFTa}. The result follows.
\end{proof}

\subsubsection{The product-derivation property}

We consider 
$\fFT_a[h(i\px)e^f](y)$
where $h(x)$ is a formal power series and
$$
f(x) = \sum_{k\ge1} \lam_k \frac{x^k}{k!}
\quad\mbox{and}\quad 
g(x) \coloneqq f(x) -  \lam_2 \frac{x^2}{2}.
$$
For convenience, we shall replace $\lam_2$ by $-a.$
From Lemma~\ref{L:ArgShift}, 
\begin{equation}\label{e:psiax}
e^{g(x)} = [w^0] \, e^{g(\pw)} e^{wx} \quad\mbox{and}\quad  h(i\px) = [z^0] \, h(\pz) e^{i z\px},
\end{equation}
so
\begin{equation}\label{e14}
\fFT_a[h(i\px)e^f](y) = [w^0 z^0] \, h(\pz) e^{g(\pw)}\,B
\end{equation}
where 
\begin{align*}
B &\coloneqq \fFT_a [e^{i z\px} e^{xw - \ha a x^2}](y) 
= \fFT_a[e^{ \ha a z^2 + i wz  + (-i a z +w)x} e^{-\ha a x^2}](y) &(\mbox{Lemma~\ref{L:TaylThm}}) \\
& = \frac{1}{\sqrt{a}} e^{\ha az^2 + i wz} \cdot e^{(-i az +w)i \py} \cdot e^{-\ha a^{-1}y^2}
& \mbox{(Def.~\ref{D:fFTa}(b))}\\  
& = \frac{1}{\sqrt{a}}  e^{-yz} \, e^{-\ha a^{-1} (y+i w)^2}  
= \frac{1}{\sqrt{a}} e^{-yz} \cdot e^{i w\py} e^{-\ha a^{-1}y^2} &  (\mbox{Lemma~\ref{L:TaylThm}}).
\end{align*}
Then, from~(\ref{e14}),
\begin{align*}
\fFT_a[h(i\px) e^f](y) 
&= \frac{1}{\sqrt{a}} \left( [z^0] h(\pz) e^{-yz}\right) \cdot
\left([w^0] e^{g(\pw)} e^{i w\py} e^{-\ha a^{-1}y^2}\right) \\
&=  \frac{1}{\sqrt{a}} h(-y) \cdot e^{g(i \py)} e^{-\ha a^{-1} y^2}
&\mbox{(from~(\ref{e:psiax}))} \\
&=   \frac{1}{\sqrt{a}} h(-y) \cdot e^{f(i \py)} e^{-\ha a\py^2} e^{-\ha a^{-1} y^2}   
= h(-y) \cdot \fFT_a[e^f](y)
&\mbox{(Def.~\ref{D:fFTa}(a))}.
\end{align*}

This proves the following property.
\begin{lemma}[ Product-derivation I]\label{ProdDeriv}
Let $f$ and $h$ be formal power series in $R[[x]]$ such that $[x^2] f(x)\neq 0$ then
$$\fFT_a[h(i\px)e^f](y) = h(-y) \cdot \fFT_a[e^f](y).$$
\end{lemma}

\subsubsection{The value on Gaussians}

The result will follow from the following result.

\begin{proposition}
\label{prop: value Gaussians a and b} 
For indeterminates $a$ and $b$ we have that $\fFT_a\left[e^{\ha b^{-1} x^2}\right](y) =  \sqrt{-b} \cdot e^{\ha b y^2}$.
\end{proposition}

\begin{proof}
By Definition~\ref{D:fFTa}(a) we have that 
\begin{align*}
\fFT_a\left[e^{\ha b^{-1} x^2}\right](y)
&= \frac{1}{\sqrt{a}} \cdot e^{-\ha (b^{-1}+a)  \py^2} \, e^{-\ha a^{-1} y^2}
\end{align*}
Next, we use Lemma~\ref{L:PQexp} for  $p=-(b^{-1}+a)$ and $q=-a^{-1}$ to obtain, 
\begin{align*}
\fFT_a\left[e^{\ha b^{-1} x^2}\right](y)&=  \frac{1}{\sqrt{a}\sqrt{-(ab)^{-1}}} \cdot e^{\ha b y^2}.
\end{align*}
\end{proof}

\begin{remark}
In the choice of the sign of the occurring square roots, care needs to be taken. We choose the signs such that our Fourier transforms matches the analytic Fourier transform on the domain where these transforms are jointly defined. Furthermore, while our Fourier transforms are also defined  on formal power series, no further choices of signs of square roots arise there. This is because then square roots such as the square root in Lemma~\ref{L:PQexp} are themselves formal power series and as such unambiguous.  
\end{remark}

\begin{lemma}[Value on Gaussians] \label{lem: value Gaussians}
$\fFT_a[e^{-x^2/2}] = e^{-y^2/2}$.
\end{lemma}

\begin{proof}
The result follows by Proposition~\ref{prop: value Gaussians a and b} with $b=-1$.
\end{proof}

\subsubsection{The quasi-involutory property}
Let $a$ and $b$ be  indeterminates.  Then
\begin{align*}
\fFT_b[\fFT_a[e^f]](z)
&= \frac{1}{\sqrt{a}} \cdot \fFT_b\left[e^{f(i \py)} e^{-\ha a \py^2} e^{-\ha a^{-1}y^2}\right] (z)
&\mbox{(Def.~\ref{D:fFTa})} \\
&= \frac{1}{\sqrt{a}} \cdot e^{f(-z)} e^{\ha a z^2} \fFT_b\left[e^{-\ha a^{-1} y^2}\right](z)
&\mbox{(Lemma~\ref{ProdDeriv})}
\end{align*}
Next, by Proposition~\ref{prop: value Gaussians a and b} we have that $\fFT_b\left[e^{-\ha a^{-1} y^2}\right](z)
=  \sqrt{a} \cdot e^{-\ha a z^2}$. Then
$$\fFT_b[\fFT_a[e^f]](z) = e^{f(-z)}.$$
Since the $\fFT_a$ and $\fFT_b$ are independent of $a$ and $b$, respectively, this proves the following result.

\begin{lemma}[Quasi-involutory]\label{QuasiInv}
The transform $\fFT$ is quasi-involutory. That is 
$\fFT[\fFT[e^f]](z) = e^{f(-z)}.$
\end{lemma}

Combining the product-derivation and quasi-involutory property we obtain the following useful version of the former.

\begin{lemma}[ Product-derivation II]\label{ProdDerivII}
Let $f$ and $h$ be formal power series in $R[[x]]$ such that $[x^2] f(x)\neq 0$ then
$$\fFT_a[h\cdot e^f](y) = h(-i \partial_y) \cdot \fFT_a[e^f](y).$$
\end{lemma}

\begin{proof}
By Lemma~\ref{QuasiInv} and Lemma~\ref{ProdDeriv} we have that 
\[
\fFT_a(h(i\partial_x) \fFT_a e^f)(y) = h(-y) \fFT_a(\fFT_a e^f)(y) = h(-y) e^{f(-y)}.
\]
Next we apply the transform $\fFT_a$ to both sides and use the quasi-involutory property, Lemma~\ref{QuasiInv}, again to obtain the desired result.
\end{proof}

\subsubsection{Translation invariance} The proof of the following is immediate from Lemma~\ref{ProdDeriv}.
\begin{lemma}[Translation invariance (up to a phase)]\label{TransInv}
Let $c$ be an indeterminate. Then
$$
\fFT_a[f(x+c)](y) = \fFT_a e^{c\px} f(x)](y) = e^{-cx} \fFT_a[f](y).
$$
\end{lemma}

\subsubsection{Proof of the Main Theorem}\label{SSS:MainThmPf}
The proof of the main theorem is now completed as follows.
\begin{proof} 
The transform
$$
\fFT_a[f](y) \coloneqq \frac{1}{\sqrt{a}}  f(i \py) e^{-\ha a\py^2}e^{-\ha a^{-1}y^2},
$$
is a formal power series and,   by Lemma~\ref{ProdDeriv} (product-derivation),  Lemma~\ref{QuasiInv} (quasi-involutory), Lemma~\ref{TransInv} (translation invariance), and Lemma~\ref{lem: value Gaussians} (value on Gaussians)  has each of the properties listed in Properties~\ref{PRPTY:Fourier}.  It is therefore an algebraic Fourier transform.
\end{proof}

\subsection{A combinatorial approach to the Fourier transform}\label{SS:CombAppFourier}
\subsubsection{The Fundamental Set and Fundamental Generating series}
Each constituent of expression~(\ref{e:expGDiffE}) has appeared in Section~\ref{SS:EnumLemmas} on enumerative lemmas,  and its combinatorial import is therefore known.  With this information, it is possible to propose a combinatorial structure of which~(\ref{e:expGDiffE}) is the generating series and to do so \emph{ab initio}, without reference to the physical background.  Indeed, 
\begin{itemize}
\item [--] the presence of $g(x)$ is suggestive of the generating series for vertices;
\item [--] the presence of $\frac{1}{2}y^2$ is suggestive of the generating series for an edge, and $e^{\frac{1}{2}y^2}$  the generating series for a disjoint union  of edges;
\item [--] the presence of $\partial_y$ and its action are suggestive of incidence of edges with vertices through the glueing operation $\Join$ defined in Section~\ref{SSS:DecompGraph}.
\end{itemize}
The object produced is a \emph{graph} (with labelling conventions to be determined).  With the introduction of~(\ref{e:expGDiffE}) and the subobject labelling it indirectly imposes, we may now be precise about the sets  $\cG^\ell$ and $\cG^{\ell,c}$ alluded to in~(\ref{CD:diagr2}).

\begin{definition}[The set $\cG^\ell$, its subobjects and weight functions]\label{D:setG}
$\cG^\ell$ is the set of all graphs with 
\begin{itemize}
\item attachment points on vertices of degree $k\ge2$ are labelled;  
\item $1$-vertices are labelled;  
\item each component has at least one edge.
\end{itemize}
\end{definition}

The notation for the corresponding generic subobjects, specified according to~(\ref{e:DefSubObj}), is:
\begin{equation*}\label{e:SubObjs}
\left\{
\begin{array}{ccl}
\sa &:=& \SOe{attach. pt. on $k$-prevertex, $k\ge2$}{x}; 	\\
\sfe &:=& \SOo{pre-edge}{u};						\\
\sfe^{-1} &:=&  \SOo{anti-edge}{u^{-1}};				\\
\sv_1 &:=& \SOe{$1$-prevertex}{y};					\\ 
\sv_k &:=& \SOo{$k$-prevertex}{\lam_k}, \quad k\ge2,
\end{array}
\right.
\end{equation*}
and $\om$ is the weight function defined by
\begin{equation*}\label{e:om}
\om := \om_{\sv_2} \ot  \om_{\sv_3} \ot \cdots
\end{equation*}

\subsubsection{The Fundamental series}
The following definition is based on the above observations.
\begin{definition}[The set $\cF$ and the series $F$]\label{D:FundSet} 
\mbox{$\quad$} \\
\begin{itemize}
\item [(a)] The \emph{Fundamental Set} is
$
\cF:=
\left\{ \Bmyputeps{0.5}{0.45}{preedge}^{-1}, \, \quad\myputeps{-0.2}{0.40}{2Vertex},\, \quad\myputeps{-1.00}{0.40}{3Vertex},\, \quad\myputeps{-1.10}{0.40}{4Vertex},\,  \ldots\: \right\}.
$
The first element is the \emph{anti-edge}; the remaining elements are \emph{internal vertices} of degrees $2, 3, 4, \ldots\,.$
\item [(b)] The term
$
E_{-}(u,x) := u^{-1}\frac{x^2}{2!}
$
is the \emph{anti-edge series},
\item [(c)] The series
$
\Fi(x):= \sum_{k\ge2} \lam_k\frac{x^k}{k!}
$
is the \emph{internal vertex series}. 
\item [(d)] The series 
$
E(u,x) := E_{-}(u^{-1},x) = u\frac{x^2}{2!}
$
is called the \emph{edge series}.
\item [(e)] The \emph{Fundamental Generating Series} is
$
F(x; u, \lam_2, \lam_3, \ldots) 
:=   \gensb{\cF}{\om_{\sa}\ot \om_{\sfe^{-1}} \ot \om}(x; -u^{-1}, \lam_2, \lam_3, \ldots).
$
Written in terms of the anti-edge series and the internal vertex series, the Fundamental Generating Series is equal to
$\Fi(x)- E_{-}(u,x).$

\item [(f)] We impose the restriction that $u\lambda_2 \neq 1$ (see Appendix~\ref{sec:exceptionalelement}).
\end{itemize}
\end{definition}

The r\^{o}le of the anti-edge has been explained already in~Remark~\ref{R:AntiEdge}.  

\subsubsection{A combinatorial transform $\bbA$}
The next result gives the generating series for this set of graphs $\cG^\ell$ as an algebraic Fourier transform.

\begin{theorem}[Main theorem for the combinatorial Fourier transform: $G_{\ell}$ and $Z_F$]\label{T:gsZF}
Let 
\begin{equation*}\label{e:ZSy}
Z_F(y) := \left. e^{\Fi (x)}\right|_{x\mapsto \partial_y}  e^{\frac 12  u y^2}.
\end{equation*}
Then
$$Z_F(y) = \gensb{\cG^\ell}{\om_\sa\ot \om_{\sv_1} \ot \om_{\sfe} \ot \om} (x=1, y; u, \lam_2, \lam_3, \ldots)
\in \bbQ[\lam_2, \lam_3, \ldots]\, [[u,y]].$$
\end{theorem}

\begin{remark} \label{rem:combtranslocallyfinite}
Note that since $Z_F(y)$ is in a ring of formal power series in the indeterminate $u$ and $y$ then its coefficients are locally finite with respect to the indeterminates $\lambda_2,\lambda_3,\ldots$, see for example \eqref{eq:master_edge} in the Appendix.
\end{remark}

\begin{proof}[Proof of Theorem~\ref{T:gsZF}]
Now
$\gensb{\{\myputeps{-0.1}{0.30}{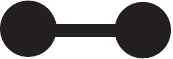}\}}{\om_{\sv_1}\ot\om_\sfe}(y;u) = \frac 12  u y^2$
so, by the Composition Lemma and~(\ref{e:Uexp}),
$$e^{\frac 12  u y^2} 
= \gensb{\cU\oc\{\myputeps{-0.1}{0.30}{Zgraph1dot}\}}{\om_{\sv_1}\ot\om_\sfe}(y;u).
$$
But there is a natural bijection  
$\cU\oc\{\myputeps{-0.1}{0.30}{Zgraph1dot}\} \bij \cM,$ 
 to the set $\cM$ of all perfect matchings (pairings), so we have made the identification 
$$
e^{\frac 12  u y^2} =  \gensb{\cM}{\om_{\sv_1}\ot\om_\sfe}(y;u),
$$
the generating series for perfect matchings with respect to vertices and edges.

For the series $e^{\Fi (x)}$, 
let $\cV_{\ge2}$ be the set of all pre-vertices of degree greater than or equal to $2$. Then, by the Composition Lemma,
$$e^{\Fi (x)} =\gensb{\cU\oc\cV_{\ge2}}{\om_\sa \ot \om}(x;\lam_2,\lam_3,\ldots).$$ 
Thus, $e^{\Fi (x)}$ is identified as the generating series for all edgeless pre-graphs with respect to pre-vertices of degree greater than or equal to $2$ with labelled attachment points.  

Now consider the set 
$$\left(\cU\oc\cV_{\ge 2}, \sa\mapsto \partial_{\sv_1}\right)$$
consisting of all operators $\hat{\fp}$ constructed from each edgeless pre-graph $\fp\in \cU\oc\cV_{\ge 3}$ by attaching $\partial_{\sv_1}$ to each attachment point ($\sa$-subobject) in $\fp.$  We define an action of $\hat{\fp}$ on the set $\cM$ of all perfect matchings as follows:
\begin{enumerate}
\item If $\om_\sa(p) > \om_{\sv_1}(\fm)$, then $\hat{\fp}(\fm)=\so$, the null subobject, whose generating series is~$0.$  The combinatorial reason for this is that there are more attachment points on $\fp$ than there are $1$-vertices in $\fm$, so $\hat{\fp}$ differentiates $\fm$ to zero.
\item If $\om_\sa(\fp) \le \om_{\sv_1}(\fm)$,  the $1$-vertices of $\fm$ with the top $\om_\sa(\fp)$ labels are converted into \emph{free-ends} retaining their labels.  Then, for $j=1,\ldots,\om_\sa(\fp)$, the free end of $\fm$ with the $j$-th highest label is joined to a pre-vertex of $\fp$  at the attachment point with the $j$-th highest label (among the attachment points).
\end{enumerate}

The resulting object has edges and no pre-vertices with unoccupied attachment points, and is therefore a graph $\fh$ in $\cG^\ell .$ It is immediately seen that the  map
\begin{equation*}\label{e:OmUVM}
\Om\colon \left(\cU\oc\cV_{\ge 2}, \sa\mapsto \partial_{\sv_1}\right) \times \cM \rar \cG^\ell \colon
(\hat{\fp},\fm) \mapsto \fh
\end{equation*} 
is reversible and therefore bijective. Then, by the Composition and Derivation Lemmas, 
\begin{equation*}\label{e:UVadMeS}
\gensb{\left(\cU\oc\cV_{\ge 2}, \sa\mapsto \partial_{\sv_1}\right)\times\cM}{\om_{\sv_1} \ot\om_\sfe\ot\om}(y;u,\lam_2,\lam_3,\ldots)
=
\left. e^{\Fi (x)}\right|_{x\mapsto \partial_y}  e^{\frac 12  u y^2}.
\end{equation*}
From this, the Composition Lemma and the bijectivity of  $\Om$, it follows that
\begin{equation*}\label{e:OmEq1}
\left. e^{\Fi (x)}\right|_{x\mapsto \partial_y}  e^{\frac 12  u y^2}
= \gensb{\cG^\ell}{\om_{\sv_1}\ot \om_\sfe \ot  \om} (y; u, \lam_2,\lam_3,\ldots),
\end{equation*}
completing the proof.
\end{proof}

\begin{example}
The  expansion of $Z_F(y)$ up to a few terms is:
\begin{align*}
\mbox{}
&  
  \left(\frac{1}{2}\,\lam_{{2}}  + \cdots \right)u  %
+ \left( \frac{3}{8}\,{\lam_{{2}}}^{2} +\frac{1}{8}\,\lam_{{4}} + \cdots \right) {u}^{2}	
+ \left( {\frac {5}{16}}\,\lam_{{4}}\lam_{{2}}+{\frac {5}{16}}\,{\lam_{{2}}}^{3}
+{\frac {5}{24}}\,{\lam_{{3}}}^{2}+\frac{1}{48}\,\lam_{{6}}  + \cdots  \right) {u}^{3} + \cdots	 \\
&\mbox{} 
+\left( \left( \frac{1}{2}\,\lam_{{3}} + \cdots \right){u}^{2}  + \left( \frac{1}{8}\,\lam_{{5}}+\frac{5}{4}\,\lam_{{3}}\lam_{{2}}  + \cdots  \right) {u}^{3}  + \cdots \right) y		\\
&\mbox{} 
+ \left(\left( \frac{3}{4}\,\lam_{{2}} + \cdots \right){u}^{2} 
+ \left( {\frac {15}{16}} \,{\lam_{{2}}}^{2}+{\frac {5}{16}}\,\lam_{{4}}  + \cdots \right) {u}^{3}
+ \left( {\frac {35}{32}}\,{\lam_{{2}}}^{3}+{\frac {7}{96}}\,\lam_{{6}}+{\frac {35}{48}}\,{\lam_{{3}}}^{2}+{\frac {35}{32}}\,\lam_{{4}}\lam_{{2}} + \cdots  \right) {u}^{4}  + \cdots \right) {y}^{2}	\\
&\mbox{} 
+ \left( \left({\frac {5}{12}}\,\lam_{{3}} + \cdots \right) {u}^{3} + \left( {\frac {7}{48}}\,\lam_{{5}}+{\frac {35}{24}}\,\lam_{{3}}\lam_{{2}} \right) {u}^{4} + \cdots \right) {y}^{3}
+ \cdots \; .
\end{align*}

The graphs in $\cG^\ell,$ with labelling removed, corresponding to the monomial $\lam_2\lam_4 u^4 y^2$ are:

\begin{center}
\includegraphics{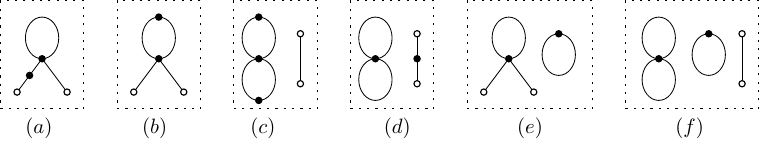}
\end{center}

The number of labellings in each case is given by:

$$
\begin{array}{crccrccr}
(a):	& \binom{6}{2} \cdot \binom{4}{3}\cdot  3! \cdot 2 = 720,	& & 	
(b):	& \binom{6}{2} \cdot \binom{4}{2}\cdot  \frac{2!}{2!} \cdot 2 = 360, 	& &	
(c):	& \binom{6}{2} \cdot \binom{4}{2}\cdot  2! \cdot 1 = 180,	\\[5pt]
(d):	& \binom{4}{2}\cdot \frac{1}{2} \cdot \binom{6}{2} \cdot 2! = 90, & &
(e):	& \binom{4}{2}\cdot \binom{6}{2} \cdot 2! = 180,	& &
(f):	& \binom{6}{2}\cdot \binom{4}{2}  \cdot\frac{1}{2} = 45.
\end{array}
$$
These totals sum to $1575$, contributing the term 
$
1575\cdot\lam_2 \lam_4 \cdot u^4 \frac{x^6}{6!} \frac{y^2}{2!} = \frac {35}{32}\cdot\lam_2\lam_4\cdot u^4 x^6 y^2
$
to $Z_F(y),$ where $x=1$.
\end{example}

\subsubsection{A combinatorial Fourier transform $\fFT$}\label{SSS:TransFourAax}
For all $e^{a(x)}$ where
$$a(x) = a_2 \frac{x^2}{2!} +  a_3 \frac{x^3}{3!} + \cdots \in R[[x]]$$
and $a_2$ is invertible in $R$, we define a transform $\fFT$ by
\begin{equation*}\label{e:Atransform}
\left(\fFT e^{a(x)}\right) (y)
\coloneqq  \left.e^{\ha a_2 x^2 -a(x)}\right|_{x\mapsto \pd{y}} e^{\ha a^{-1}_2 y^2}.
\end{equation*}
In particular, it follows from Theorem~\ref{T:gsZF} that
\begin{equation}\label{D:FTrans3}
\left(\fFT e^{F(u,x)}\right) (u,y)
\coloneqq \left. e^{F(u,x)+E_{-}(u,x)}\right|_{x\mapsto\pd{y}} e^{E(u,y)},
\end{equation}
and we note that
\begin{equation*}\label{e:FTrans1}
\left. e^{\Fi (x)}\right|_{x\mapsto \partial_y}  e^{\frac 12  u y^2} \in \mathbb{Q}[y; \lambda_2,\lambda_3,\ldots]\,[[-;u]].
\end{equation*}

\begin{remark}\label{R:CombA}
The combinatorial interpretation from Theorem~\ref{T:gsZF} for the right hand side of~(\ref{D:FTrans3}) implies that~$\fFT$ can be viewed as a combinatorial transform. For comparisons of related expressions in \cite{Bo} for the Fourier transform see Section~\ref{subsec: Borinsky} and Appendix~\ref{S:appendix_explicit_examples}.
\end{remark}

\begin{remark}
The transform $\fFT$ is very similar to the algebraic Fourier transform $\fFT_a$ from Definition~\ref{D:fFTa} where $a=-a_2$. In the latter, we do the  substitution $x \mapsto i\partial_y$ instead.
\end{remark}

The next result, although trivial, is included for completeness.
\begin{definition}[The set $\cG^{\ell,c}$]\label{D:Wseries}
$\cG^{\ell,c}$ denotes the set of all connected graphs in $\cG^\ell$. 
\end{definition}

\begin{lemma}[$G_{\ell,c}$ and $W_F$]\label{e:connFeyn}
Let
$$W_F(y; u,\lam_2,\lam_3,\ldots) := \gensb{\cG^{\ell,c}}{\wt{\sv_1}\ot\wt{\sfe}\ot \om }(y; u, \lam_2, \lam_3,\ldots).$$
Then 
$$Z_F(y) = e^{W_F(y)}.$$
\end{lemma}
\begin{proof}
Clearly,
$\cU \oc \cG^{\ell,c} \bij \cG^\ell$.
The result follows from Theorem \ref{T:gsZF} and~(\ref{e:Uexp}). 
\end{proof}

This completes the explanation of Diagram~(\ref{CD:diag12}) in purely combinatorial terms.

\subsubsection{The combinatorial quasi-involutory property of $\fFT$}\label{SS:QIpFT}
We shall prove that $\fFT$, like the (analytic) Fourier Transform, is a quasi-involution (cf. Lemma~\ref{QuasiInv} ).
\begin{theorem}\label{T:QuasiInvFT} Let $F(x)$ be the fundamental series then $\left(\fFT (\fFT e^{F})\right)(z) = e^{F(z)}.$
\end{theorem}
\begin{proof}
In the notation of~(\ref{D:FTrans3}), 
$$F(x, u, \lam_2, \lam_3, \ldots) :=   \gensb{\cF}{\om_{\sa}\ot \om_{\sfe^{-1}} \ot \om}(x; u^{-1}, -\lam_2, -\lam_3, \ldots)$$ 
is the generating series for the Fundamental Set $\cF.$ Then by Theorem~\ref{T:gsZF}
\begin{eqnarray*}
\left(\fFT e^F\right)(y) 
=  \left. e^{\Fi (x)}\right|_{x\mapsto \partial_y}  e^{\frac 12  u y^2} 
= \gensb{\cG_\cF^\ell}{\wt{\sv_1}\ot \wt{\sfe} \ot \om}(y; u, \lam_2, \lam_3, \ldots) 
= Z_F(y).
\end{eqnarray*}
But, from Lemma~\ref{e:connFeyn},
$Z_F(y) = e^{W_F(y)}$
where
$$W_F(y; u,\lam_2,\lam_3,\ldots) := \gensb{\cG^{\ell,\mathsf{conn}}}{\wt{\sv_1}\ot\wt{\sfe}\ot \om }(y; u, \lam_2, \lam_3,\ldots).$$
We regard $W_F(y)$ as the generating series for the set:
$$
\cW_\cF:= \raisebox{-5mm}{\includegraphics{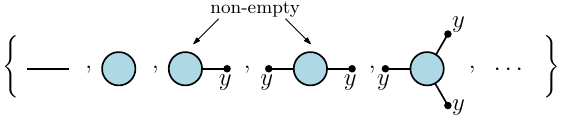}}
$$
where the diagram with $k$ 1-vertices (marked by $y$) is the set of all graphs in $\cG_\cF^{\ell,\mathsf{conn}}$ with $k$ 1-vertices, and at least one edge since $\myputeps{-0.1}{0.30}{Zgraph1dot}$ is to be excluded from \myputeps{-2}{0.12}{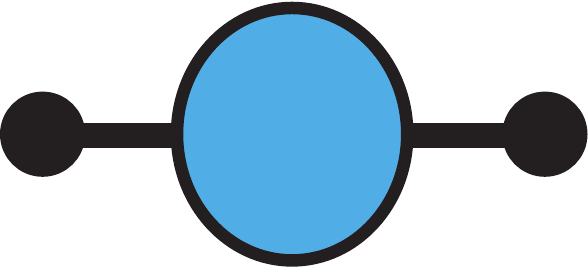}. Indeed,
$$
W_F(y) = u\frac{y^2}{2} + G(y)
\quad\mbox{where}\quad
G(y) := \sum_{k\ge0} W^{(k)} \frac{y^k}{k!}.
$$
where $W^{(k)}$ is the generating series  for the graphs in $\cW_\cF$ with $k$ 1-vertices. Moreover, the series for the set  
$$
\includegraphics{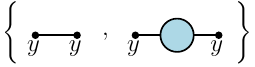}
$$
of connected graphs with two 1-vertices is
$$W^{(2)} \frac{y^2}{2} -u\frac{y^2}{2}.$$
That is, we may view the series $W_F$ as an instance of $F$ and so
$$
\left(\fFT e^{W_F}\right) (z) = 
\left. e^{G(y)}\right|_{y\mapsto \pd{z}} e^{-\ha u^{-1}z^2}.
$$

Let  $\cG_{\cW_\cF}^\ell$ be the set of all graphs with labelled open ends of an anti-edge constructed from $\cW_\cF.$  Then, by Theorem~\ref{T:gsZF},
$$
\left(\fFT e^{W_F}\right) (z)
=\gensb{\cG_{\cW_\cF}^\ell}{\wt{\sfq} \ot \wt{\sfe} \ot w_{\sv_1} \ot \cdots} (z; u^{-1}, W_F^{(1)}, W_F^{(2)}, \ldots),
$$
where, in the notation of~(\ref{e:subobjGen}),
\begin{equation*}\label{e:subobjGen2}
\left\{
\begin{array}{ccl}
\sfq &:=& \SOe{open anti-edge end}{z},		\\
\sfe &:=& \SOo{anti-edge end}{u^{-1}},		\\
\sv_k &:=&  \SOo{the set \mbox{$\cW^{(k)}_\cF$}}{W^{(k)}_F}  \quad\mbox{(note change of notation)}.
\end{array}
\right.
\end{equation*}
To establish the quasi-involutory property 
$$
\left(\fFT e^{W_F}\right)(z) = \left( \fFT(\fFT e^F)\right) (z) = e^{F(z)}
$$
of $\fFT$ we must show that
$$
\gensb{\cG_{\cW_\cF}}{\wt{\sfq} \ot \wt{\sfe} \ot w_{\sv_1} \ot \cdots} (z; u^{-1}, W_F^{(1)}, W_F^{(2)}, \ldots)
=\gensb{\cU\oc\cF}{\wt{\sa} \ot \wt{\sfe} \ot \om} (z; u, \lam_2, \lam_3, \ldots).
$$
This will be done by constructing a sign-reversing involution. To this end, let
$\ol{\cG^{\ell_,\pm}_\cF}$ be the set of all graphs in $\cG_\cF^\ell$ with:
\begin{itemize}
\item [--] each external edge is deleted, but the label of the deleted $1$-vertex is retained at the attachment point and is weighted by $-1$; 
\item [--] each internal edge is weighted by $\pm1$, in all possible ways.
\end{itemize}

Let 
\begin{equation}\label{e:GWF2GF}
\phi \colon \cG^\ell_{\cW_\cF} \rar \ol{\cG^{\ell_,\pm}_\cF}\colon \fa \mapsto \fb
\end{equation}
be a map constructed with the following element action:

\myss{For internal edges}

\begin{center}
    \includegraphics[scale=0.95]{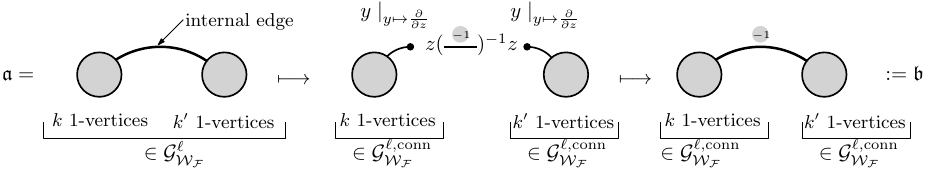}
\end{center}

\mbox{ }\\[2pt]
The two gray components in~$\fa$ may be the same component.

\myss{For external edges}

\begin{center}
    \includegraphics[scale=0.95]{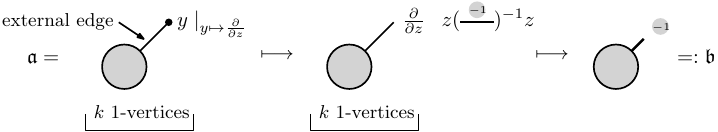}
\end{center}

Conversely, given $\fg\in\ol{\cG^{\ell_,\pm}_\cF}\colon \fa \mapsto \fb$, each internal edge weighted by $-1$ is cut, to give $t$ components $\fg_1, \ldots, \fg_t$ with $i_1, \ldots, i_t$ $1$-vertices, respectively.  Thus $\phi$ is bijective.

Let $\cB$ be the subset of objects in $\ol{\cG^{\ell_,\pm}_\cF}$ with at least one edge, and let
$$\Om\colon \cB \rar \cB \colon \fg \mapsto \fg'$$
where $\fb'$ is constructed from $\fg$ by reversing the sign of the multiplicative weight on the internal edge of $\fg$ incident with the attachment point with smallest label. Then 
$\Om$ is a sign-reversing involution; 
$\ol{\cG^{\ell_,\pm}_\cF} - \cB = \cU\oc\cF$, the fundamental series
given in Definition~\ref{D:FundSet}, that is, all graphs in $\ol{\cG^{\ell_,\pm}_\cF}$ with no internal edges.
By the sign-reversing involution lemma,
$$
\gensb{\ol{\cG^{\ell_,\pm}_\cF}}{\wt{\sa} \ot \wt{\sfe} \ot \om} (z; u, \lam_2, \lam_3, \ldots)
=
e^{\ha u^{-1} z^2 - \sum_{k\ge2}  \frac{1}{k!} \lam_k z^k}
$$
and so, from the bijection~(\ref{e:GWF2GF}),
$$
\gensb{ \cG^\ell_{\cW_\cF}}{\wt{\sa} \ot \wt{\sfe} \ot \om} (z; u, \lam_2, \lam_3, \ldots)
=
e^{\ha u^{-1}  z^2 - \sum_{k\ge2} \frac{1}{k!} \lam_k z^k}.
$$
This completes the proof.
\end{proof}

\subsubsection{The combinatorial product-derivation property of $\fFT_a$} \label{subsec:combproofproduct-derivation}

The main result of this section is the combinatorial proof of the
product-derivation property (Lemma~\ref{ProdDeriv}).

We need the following result.

\begin{proposition} \label{prodder:cor1}
Let $f(x) = \frac12 c x^2 + g(x)$ where $c\neq 0$ and $r(x)$ be a polynomial, then
\[
\mathbb{F}_a[r(x)e^{f}](y) = \frac{1}{\sqrt{c}} \cdot r(i\partial_y)
e^{g(i\partial_y)} e^{\frac12 c^{-1} y^2}
\]
\end{proposition}

\begin{proof}
We apply
Lemma~\ref{L:PQexp} with $p=a-c$ and $q=a^{-1}$ to
$\mathbb{F}_a[re^f](y)$. Since $1-pq = a^{-1}c$ we obtain
\begin{align*}
\mathbb{F}_a[re^{f}](y) &:= \frac{1}{\sqrt{a}} \cdot
r(i\partial_y)e^{g(i\partial_y)} e^{\frac12(a-c)\partial_y^2}e^{\frac12 a^{-1} y^2}\\
&= \frac{1}{\sqrt{c}} \cdot r(i\partial_y)e^{g(i\partial_y)} e^{\frac12 c^{-1} y^2},
\end{align*}
as desired.
\end{proof}

\begin{corollary} \label{prodder:cor2}
Let $F(x) = u^{-1}\frac{x^2}{2!} + F_{\rm int}(x)$, where $F_{\rm
  int}(x) = \sum_{k\geq 2} \lambda_k \frac{x^k}{k!}$, then
\[
\mathbb{F}_{u^{-1}}[e^{F(x)}](y) = \sqrt{u}\cdot   [[\mathcal{G}^{\ell}, \omega_a \otimes \omega_{v_1} \otimes
\omega_e \otimes \omega]](x=1,y,u,\lambda_2,\lambda_3,\ldots).
\]
\end{corollary}

\begin{proof}
The result follows by combining Proposition~\ref{prodder:cor1} for
$r(x)=1$ and Theorem~\ref{T:gsZF}.
\end{proof}

\begin{proof}[Combinatorial proof of Lemma~\ref{ProdDeriv}]
By linearity of $\mathbb{F}_a[\cdot]$ it suffices if we prove the case
when $h(x) = x^k$ for a nonnegative integer $k\geq 0$:
\begin{equation} \label{eq1}
\mathbb{F}_a[ (i\partial_x)^ke^f](y) = (-y)^k\cdot \mathbb{F}_a[e^f](y).
\end{equation}

Recall that $e^{f(x)} = e^{g(x)} e^{\frac12 cx^2}$. By the derivation and composition rule
\[
(i\partial_x)^k e^{g(x)}\cdot e^{\frac12 c x^2} =
[[ \partial_{\mathsf{a}}^k \left(\mathcal{U} \circ
( \mathcal{E}^{-1} \cup \mathcal{V} )\right), \omega_{\mathsf{a}} \otimes
\omega ]](x; c, \lambda_2, \lambda_3,\ldots)
\]
is the generating series for all edgeless pre-graphs with respect to
pre-vertices and inverse edges with labelled attachment points and $k$ deleted
attachment points. Moreover by the derivation lemma and the product rule
\[
(i\partial_x)^k e^{g(x)}\cdot e^{\frac12 c x^2} = p_k(x)\cdot e^{g(x)}\cdot e^{\frac12 c x^2},
\]
where $p_k(x)$ is the generating polynomial accounting for the objects with exactly $k$
deleted points of attachment distributed between vertices and inverse edges. For example for $k=1,2$, $p_1(x)$ and
$p_2(x)$ are given by
\begin{align*}
p_1(x) &= i\partial_x g(x) + icx,\\
p_2(x) &= -\partial^2_x g(x) - (\partial_x g(x) )^2 - 2cx \cdot \partial_x
g(x) - c - c^2x^2.
\end{align*}


By Proposition~\ref{prodder:cor1}
\begin{equation} \label{eq:gsgraphs}
\mathbb{F}_a [ (i\partial_x)^ke^{f(x)}](y) = \frac{1}{\sqrt{c}}
p_k(i\partial_y) e^{g(i\partial_y)} e^{\frac12 c^{-1} y^2}.
\end{equation}
By the construction in the proof of Theorem~\ref{T:gsZF}, \eqref{eq:gsgraphs} is
the generating series for graphs with labelled $1$-vertices and $k$
deleted attachment points. Note any two of these deleted attachment points could be
paired by an inverse edge: \includegraphics{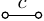}. The deleted attachment points are weighted
by $\pm 1$ and they arrive in two possible ways.

\begin{itemize}
\item[(i)] From an inverse edge with one deleted point of
  attachment. That is, from a term
  $(i\partial x) \frac12 c x^2 = i c x$, glued with the
  substitution $x=i\partial_y$ to a $1$-vertex of a graph. See
  Figure~\ref{fig:delattachpoint} top left. This gives a deleted attachment point weighted by
  $-1$ signaling that it came from an
  edge contraction. If the $1$-vertex came from a component of a graph
  consisting of one edge then we call the resulting $1$-vertex
  with a deleted attachment point {\em isolated}. See
  Figure~\ref{fig:delattachpoint} bottom left.

\item[(ii)] From a vertex with one attachment point deleted or from an
  inverse edge with both attachment points deleted that remains so
  after applying the transform $\mathbb{F}_a$. See Figure~\ref{fig:delattachpoint}
  right. This gives a deleted attachment point weighted by $1$
  signaling that it did not come from an edge contraction. Note that
  in this case  isolated deleted attachement points do not arise. 
\end{itemize}

\begin{figure}
\begin{center}
\includegraphics{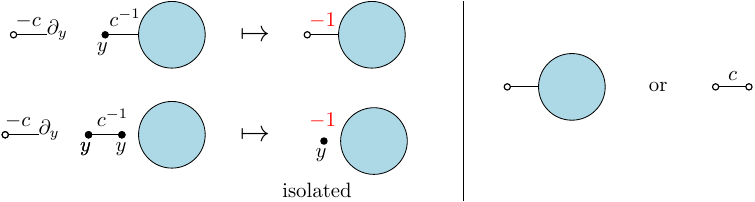}
\caption{Types of deleted attachment points: from an edge contraction
  which includes isolated deleted attachment points, without an edge
  contraction.}
\label{fig:delattachpoint}
\end{center}
\end{figure}


Let $\mathcal{G}^{\ell,\circ \pm}_k$ be the
set of graphs enumerated in \eqref{eq:gsgraphs}. Let $\mathcal{B}$ be
the subset of objects in $\mathcal{G}^{\ell,\circ \pm}_k$ with at
least one non isolated deleted attachment point, and let 
\[
\Omega: \mathcal{B}\to \mathcal{B}: \mathsf{g} \mapsto \mathsf{g}',
\]
where $\mathsf{g}'$ is constructed from $\mathsf{g}$ by reversing the
sign of the multiplicative weight of the deleted attachment point with
the smallest label. Then $\Omega$ is a sign-reversing involution and
the complement
\[
\mathcal{G}^{\ell,\circ \pm}_k - \mathcal{B} = \mathsf{v_1}^k \times
\mathcal{G}^{\ell},
\]
consists of graphs with $k$ isolated deleted attachment points. By the
sign-reversing involution Lemma,
\begin{align*}
\mathbb{F}_a [ (i\partial_x)^ke^{f(x)}](y) &= [[\mathcal{G}^{\ell,\circ
  \pm}_k, \omega_{\mathsf{a}} \otimes \omega_{\mathsf{e}} \otimes
  \omega]](z; c^{-1}, \lambda_2,\lambda_3,\ldots) \\
&= (-y)^k \cdot  [[\mathcal{G}^{\ell}, \omega_{\mathsf{a}} \otimes \omega_{\mathsf{e}} \otimes
  \omega]](z; c^{-1}, \lambda_2,\lambda_3,\ldots).
\end{align*}
Finally by Theorem~\ref{T:gsZF} the RHS equals $(-y)^k\cdot
\mathbb{F}_a[e^f](y)$ as desired.
\end{proof}

\vspace{10pt}

\section{The Formal and diagrammatic Legendre Transform: univariate case}\label{S:FormalLT1}

This section gives an expanded treatment of the Legendre transform than that in \cite{JKM1} including the new result of the combinatorial proof of the quasi-involutory property of the Legendre transform (Theorem~\ref{SS:QIpLT}). The following routine facts with be useful.  If $a,b\in R[[x]]$ have compositional inverses, then $b\oc a$ does also and $\mi{(b\oc a)} = \mi{a}\oc\mi{b}.$  Also, if $\lam\in R$ is invertible then $\mi{(\lam a)} = \mi{a} \lam^{-1}$ since $\mi{\lam} = \lam^{-1}$ (regarding $\lam$ here as product by $\lam$). We note that  compositional inversion is involutory on $R[[x]],$ a property that will be seen to be inherited by the combinatorial Legendre Transform. 

It will be convenient to denote the derivative of $a$ with respect to $x$ by $\diff{a}$.

\subsection{An algebraic Legendre transform}
\subsubsection{Basic properties}\label{SSS:BasicPropFLT}
It is necessary first to establish some of properties the combinatorial Legendre Transform $\fLT$ shares with the (analytic) Legendre Transform $\LT$.

\begin{definition}[Univariate Combinatorial Legendre Transform $\fLT$]\label{D:fLT}
\textit{
If $a\in R[[x]]$ and $\ds{\dmi{a}(x)}$ exists, then the \emph{Combinatorial Legendre Transform} of $a(x)$ is
$$(\fLT a)(x) := \left(a\oc \dmi{a}\right)(x) - x \cdot \dmi{a}(x).$$
}	
\end{definition}

The next result give necessary and sufficient conditions for $\fLT a$ to exist.

\begin{proposition}\label{P:Exist:fLT}
$a(x) \in R [[x]]$ has a Legendre transform in $R [[x]]$ if and only if 
$$(i) \quad [x^0] \diff{a}(x) = [x] a(x)=0 \qquad\mbox{and}\qquad  (ii)\quad[x^2]a(x)\neq 0.$$ 
\end{proposition}
\begin{proof}
Let $a\in R[[x]].$ Then, from Definition~\ref{D:fLT}, $\fLT a$ exists if and only if $\dmi{a}$  and $a\oc \dmi{a}$ exist. The first exists if and only if $[x^0]\diff{a}=0$ and $[x]\diff{a}\neq0$, and the second then exists if and only if  $[x^0]a=0.$  The proof then follows.
\end{proof}

The next result expresses $(\fLT a)(x)$, when it exists,  as an explicit series in $x$.
\begin{proposition}\label{P:xFLTy}
If $a\in R[[x]]$ and $\dmi{a}(x)$ exists, then
$\ds{\diff{(\fLT a)} (x) =  - \dmi{a}(x).}$  
\end{proposition}
\begin{proof} We differentiate and apply the chain rule to obtain
\[
\ds{\diff{(\fLT a)} (x) = \left(\diff{a}\oc \dmi{a} \right) (x)\cdot(\diff{\dmi{a}}(x))  
- x \cdot(\diff{\dmi{a}}(x)) - \dmi{a}(x) = - \dmi{a}(x)}.
\]
\end{proof}

The next result shows that~$\fLT$ is also quasi-involutory.

\begin{lemma}\label{L:qInvLT}
If $a\in R[[x]]$ has a combinatorial Legendre Transform then 
$(\fLT^2 a)(-x)=a(x).$  
That is, $\fLT$ is quasi-involutory where it is defined.
\end{lemma}
\begin{proof}
We assume that $a\in R[[x]]$ has a combinatorial Legendre Transform.  Then, by Proposition~\ref{P:Exist:fLT},  $[x^0]a = [x]a =0$ and $[x^2]a \neq 0.$  Thus $a(x) = \sum_{k\ge2} a_k x^k$ where $a_k\in R$ for $k\ge2$ and $a_2\neq 0.$  It is readily seen that
$\ds{(\fLT a)(x) = -(2a_2)^{-1} x^2 + \cdots}$, so $(\fLT^2a)(x)$ exists by Proposition~\ref{P:Exist:fLT}.
Then
$(\fLT^2a)(-x) = ( (\fLT a) \oc  \dmi{(\fLT a)}) (-x) + x\cdot \dmi{(\fLT a)}(-x)$
from Definition~\ref{D:fLT}. But, from Proposition~\ref{P:xFLTy},  
$\dmi{(\fLT a)} (-x) =  \mi{(-\dmi{a})}(-x) = \diff{a}(x)$ 
so
\begin{eqnarray*}
(\fLT^2a)(-x) 
=  (\fLT a) (\diff{a}(x)) + x\cdot \diff{a}(x) 
=   ( a\oc\dmi{a}\oc \diff{a})(x) 
- \diff{a}(x) \cdot(\dmi{a}\oc \diff{a})(x) + x\cdot \diff{a}(x)
\end{eqnarray*}
where the latter is again from Definition~\ref{D:fLT}.  The result follows immediately.
\end{proof}

The next result is an immediate consequence of the quasi-involutory property.

\begin{proposition}\label{P:aEb}
If $a,b\in R[[x]]$ have Legendre Transforms and $\fLT a = \fLT b$ then $a=b.$
\end{proposition}

\subsubsection{An explicit expression for  $(\fLT a)(x)$}
An explicit formal power series presentation of the combinatorial Legendre Transform is given in the next result. 
\begin{corollary}\label{condlt}
Suppose $a\in R[[x]]$ and $\fLT a$ exists, so $\diff{a}(x) =x\cdot h(x)$ and $h(0)\neq0.$  Then  
\begin{equation*} \label{exeqlt}
(\fLT a)(y) = - \sum_{k\geq 1} \frac{1}{k(k+1)} y^{k+1} [\lam^{k-1}] h^{-k}(\lam).
\end{equation*}
\end{corollary}

\begin{proof}
From Proposition~\ref{P:Exist:fLT},  we may assume that $[x^0]\diff{a}=0$ and $[x]\diff{a}\neq0.$ Therefore, indeed, $h(0)\neq0$ so $h$ is invertible.  Then
$\dmi{a}(y) = \sum_{k\ge1} \frac1k y^k [t^{k-1}] h^{-k}(t)$ from Corollary~\ref{C:CompInv}, and the result follows from Proposition~\ref{P:xFLTy}.
\end{proof}

We give an application of this explicit expression for the Legendre transform.

\begin{example}[{\cite[Example 6]{JKM1}}]
Let $T_F(y)$ be $\ds{\gensb{\cT^\ell}{\om_{\sv_1}\ot\om_{\sfe}\ot\om}(y;u, \lam_2,\lam_3,\cdots)}$. If $u=1$ and $\lam_i=(i-3)!$ By Lemma \label{L:framedT} we have that $(\mathbb{L} F)(y) = -T_{F}(-y)$. We can compute this transform explicitly using Corollary \ref{condlt}.
\[
\left. (\mathbb{L} F)(y) \right|_{u=1,\lam_i=(i-3)!} = - \sum_{k\geq 1} \frac{1}{k(k+1)} y^{k+1} \left[t^{k-1}\right] \left. (\diff{F}(t)/t)^{-k}\right|_{u_1,\lam_i=(i-3)!}\\
\]
But $\left. \diff{F}(t) \right|_{u=1,\lam_i=(i-3)!} = (t-1)\log(1-t)$ so
\begin{eqnarray*}
\left. T_{F}(-y) \right|_{u=1,\lam_i=(i-3)!} 
&=  \sum_{k\geq 1} \frac{1}{k(k+1)} y^{k+1} \left[t^{k-1}\right] (t-1)^{-k} (\log(1-t))^{-k}t^k	\\
&= - \sum_{k\geq 1} \frac{1}{k(k+1)} y^{k+1} \left[t^{-1}\right] (t-1)^{-k}(\log(1-t))^{-k}.
\end{eqnarray*}
Let $r(t) = 1 - e^t$. Then $\val(r)=1.$ Thus, by the Residue Composition Theorem~\ref{T:ResComp},
$$
\left[t^{-1}\right] (t-1)^{-k}(\log(1-t))^{-k}
= (-1)^k \left[t^{-1}\right]e^{-(k-1)t} t^{-k} =   -  \frac{(k-1)^{k-1}}{(k-1)!},
$$
and we obtain that 
\[
\left. T_F(y)\right|_{u=1,\lam_i=(i-3)!} = \sum_{n \geq 3} (n-2)^{n-2} \frac{y^n}{n!}.
\]
\end{example}

\subsection{A combinatorial approach to the Legendre Transform}\label{S:Eff:Act1} 
This section proves by combinatorial constructions \textbf{(i)} the relationship between trees and the combinatorial Legendre transform $\fLT$, \textbf{(ii)} the relation between connected graphs and edge 2-connected graphs, and \textbf{(iii)} the quasi-involutory property of $\fLT$.

\newcommand{\phan}{\phantom{n}}

\subsubsection{The relationship between trees and $\fLT$}

\begin{definition}[The set $\cT^\ell$]
$\cT^\ell \subset \cG^\ell$ is the set of all trees with labelled $1$-vertices and labelled attachment points.
\end{definition}

The next result establishes a connexion between the combinatorial Legendre transform~$\fLT$ and trees.  In the statement of the result, it is necessary to allow the sum over $k$ to start at $1.$ This therefore means that there are two types of $1$-vertices:
\begin{itemize}
\item[--]  those marked by the indeterminate $y$ (exponential);
\item [--] those marked by the indeterminate $\lam_1$ (ordinary).
\end{itemize}

\begin{theorem}[Tree property of $\fLT$; \cite{JKM1}]\label{L:TreeDiff}
Let
$T_F(y) := \gensb{\cT^\ell}{\om_{\sv_1}\ot\om_{\sfe}\ot\om}(y;u, \lam_2,\lam_3,\cdots)$
where \\
 $\ds{F(x) := -u^{-1} \frac{x^2}{2!} + \sum_{k\ge2} \lam_k \frac{x^k}{k!}}$.
Then 
$(\fLT F)(y) = T_F(-y).$
\end{theorem}
\begin{proof}
We shall use four additively $\om_{\sv_1} \ot \om_\sfe \ot \om$-preserving bijections involving the set $\cT^\ell$.

\textbf{First bijection - Partition of Unity:} 
By the Euler-Poincar\'{e} Theorem,
\begin{equation}\label{e:PtnUnityA}
1 = \om_\sv(\ft) - \om_\sfe(\ft)
\end{equation}
where $\ft$ is a tree. But $\om_\sv(\ft)$ and $\om_\sfe(\ft)$ weight each tree by its number of vertices and edges, respectively. Counting trees with respect to the first weight may be achieved by distinguishing a single vertex in each tree in all possible ways, and similarly for edges. Written in the form $1 + \om_\sfe(\ft) =  \om_\sv(\ft)$, the above relation~(\ref{e:PtnUnityA})  implies that
\begin{equation}\label{e:decomp1}
\cT^\ell \uplus \sfe  \pe \cT^\ell \bij \biguplus_{k\ge1} \sv_k \pvk{k} \cT^\ell.
\end{equation}
Each of the sets $\sfe  \pe \cT^\ell$ and $\sv_k \pvk{k} \cT^\ell$ is decomposed further, as follows.

\textbf{Second bijection - Distinguished Vertex:} 
By distinguishing a vertex of degree~$k$, 
\begin{equation}\label{e:decomp2}
\sv_k \pvk{k} \cT^\ell \bij \sv_k  \oc_{\bowtie}\left(\pvk{1}\cT^\ell\right).
\end{equation}
where $\oc_{\bowtie}$ indicates that the composition is by `glueing' to each of the $k$ attachment points of $\sv_k$ the open end of the unique pre-edge in each of $k$ tree from $\pvk{1}\cT^\ell$ (the operator $\pvk{1}$ deletes a vertex of degree one, leaving a single attachment point).

\textbf{Third bijection - Edge Deletion:}   The combinatorial operator $\sfe^{-1}\pvk{1}$ deletes a $\sv_1$-subobject together with its incident pre-edge.  Then

\begin{center}
    \includegraphics[]{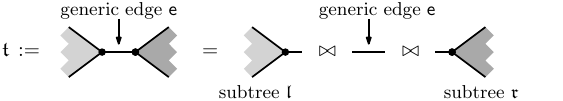}
\end{center}

where the two subtrees are distinguished by their vertex labels. Therefore, by deleting $\sfe$, we have
$\pe \ft  = \{\fl, \fr\}$
(the order of $\fl$ and $\fr$ being immaterial).
But
$$\fl  
= \sfe^{-1}\pvk{1}  \fl'
\quad\mbox{where}\quad
\fl'  = \raisebox{-5pt}{\includegraphics[]{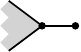}} \in \cT^\ell,
$$
$$
\fr = \sfe^{-1}\pvk{1}  \fr'
\quad\mbox{where}\quad
\fr'  = \raisebox{-5pt}{\includegraphics[]{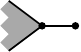}} \in \cT^\ell.
$$
Thus 
\begin{equation}\label{e:decomp3}
\pe \cT^\ell \bij \cU_2 \oc \left(\sfe^{-1}\pvk{1}  \cT^\ell\right) \colon \ft \mapsto \{\fl,\fr\}, 
\end{equation}  
which is clearly bijective.
Substituting~(\ref{e:decomp2}) and~(\ref{e:decomp3}) into~(\ref{e:decomp1}) gives the additively $\om_{\sv_1} \ot \om_\sfe \ot \om$-preserving bijection  
$\sv_1 \pvk{1} \cT^\ell \uplus \biguplus_{k\ge2} \sv_k \oc\left(\pvk{1}\cT^\ell\right)
\bij \cT^\ell \uplus
\left(\sfe \boxt \left(\cU_2 \oc \sfe^{-1}\pvk{1} \cT^\ell\right)\right).$
Applying the Sum, Product, Composition and Derivation Lemmas to this gives 
$$
y\prtl{T_F}{y} + \sum_{k\ge2} \frac{\lam_k}{k!} \left(\prtl{T_F}{y}\right)^k = T_F + \frac{u}{2} \left(\frac{1}{u}
\prtl{T_F}{y}\right)^2,
$$
having noted that $\ds{\sfe^{-1}\pvk{1} \mapsto u^{-1} \prtl{}{y}}.$  
Therefore $T_F(y)$ satisfies the formal partial differential equation
\begin{equation}  \label{e:FforT}
T_F(y) = y\cdot \diff{T_F}(y) + (F\oc \diff{T_F})(y).
\end{equation}

\textbf{Fourth bijection - Cancellation:}
The final bijection transforms the right hand side of~(\ref{e:FforT}) as the combinatorial Legendre transform of $F$, as follows. 
Let ${\cT'}^{,\ell}$ be the subset of trees in $\cT^\ell$ with at least one vertex of degree at least $3$.  Then 
${\cT'}^{,\ell} = \cT^\ell -
\left\{\myputeps{0.1}{0.30}{Zgraph1dot}\right\}.
$
Let $\ft\in {\cT'}^{,\ell}$. Then $\fs \coloneqq  \Bmyputeps{0.4}{0.45}{preedge}^{-1} \oc \pvk{1} \ft$ is obtained as follows:

\begin{center}
    \includegraphics[]{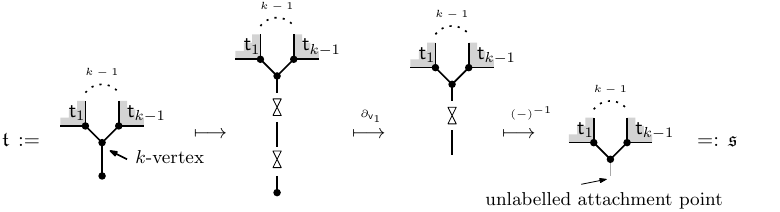}
\end{center}

where $\fs \in \Bmyputeps{0.4}{0.45}{preedge}^{-1} \oc \pvk{1} {\cT'}^{,\ell}.$

But $\fs$ may also be constructed as follows:

\begin{center}
    \includegraphics[]{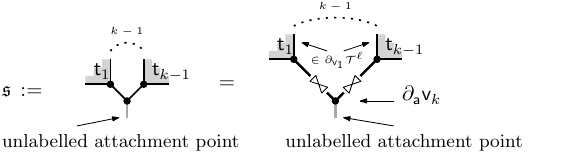}
\end{center}

This gives the element action of a bijection:
$
\Bmyputeps{0.4}{0.45}{preedge}^{-1} \oc \pvk{1} {\cT'}^{,\ell} \bij
\biguplus_{k\ge3} \left(\partial_\sa \sv_k\right) \oc \left(\pvk{1} \cT^\ell\right).
$
Then, by the enumerative lemmas,
$$
u^{-1}\left(\frac{d}{dx} \frac{x^2}{2!}\right) \oc \frac{\partial}{\partial y}\left(T_F - u\frac{y^2}{2!}\right)
=
\sum_{k\ge3} \lam_k \left(\frac{d}{d x} \frac{x^k}{k!}\right)\oc \diff{T_F}(y)
$$
whence
$$
-y =  \frac{\partial}{\partial x}\left(-u^{-1}  \frac{x^2}{2!} + \sum_{k\ge3} \lam_k \frac{x^k}{k!}\right) \oc \diff{T_F}(y)
=
\left(\diff{F}\oc \diff{T_F}\right) (y),
$$
and so
$\diff{T_F}(y) = \diff{F}^{[-1]}(-y).$
The result now follows by substituting this expression for $\diff{T_F}(y)$ into~~(\ref{e:FforT}) using Definition~\ref{D:fLT} to identify the resulting right hand side as the combinatorial Legendre transform of $F$.
\end{proof}

\subsubsection{A summary of the action of~$\fLT$}
The following diagram summarises the change of variables in the action of the combinatorial Legendre transform: 

\begin{equation*}\label{e:SummaryActfLT}
\begin{tikzcd}
F(x) \arrow[bend right=60,swap]{d}{x \,\mapsto\, \frac{\partial T_F}{\partial y}} \arrow[bend right=60]{d}{\mathbb{L}} \\
T_F(y) \arrow[bend right = 60, swap]{u}{y \,\mapsto\, \frac{\partial F}{\partial x}} \arrow[bend right = 60]{u}{\mathbb{L}}
\end{tikzcd}
\end{equation*}
where $T_F(y)$ is the generating series in $y$ for the set of all trees in $\cG^\ell$ constructed from the constituents specified in the Fundamental Series $F(x)$ and where the indeterminates $x$ and $y$ are related through $y\mapsto \frac{\partial F}{\partial x}.$

\begin{remark}[Partition of Unity]
In view of the significance of~(\ref{e:PtnUnityA}) in the present context, we shall refer to it as a \emph{Partition of Unity}. 
\end{remark}

\subsubsection{The involutory property of the combinatorial Legendre transform}\label{SS:QIpLT}
\begin{theorem}[Quasi-involutory Property]\label{T:QuasiInvolLT}
In the notation of Theorem~\ref{L:TreeDiff},
$
\left(\fLT (\fLT F)\right)(-y) = F(y).
$
\end{theorem}
\begin{proof}
For $T_F(y)$ defined in Theorem~\ref{L:TreeDiff}, let
\begin{equation}\label{e:TFk}
T^{(k)}_F := \left.\frac{\partial^k}{\partial y^k}  T_F(y)\right|_{y=0}.
\end{equation}
Then $T^{(k)}_F$ is the generating series for all trees in $\cT^\ell$  with
each of its $k$ $1$-vertices deleted, but their labels retained and no vertices of degree 2.
This set of trees is to be denoted by $\overline{\cT}^{\,(k)}_F,$ and the set of all such trees for all $k$ is denoted by $\overline{\cT}_F.$

From Theorem~\ref{L:TreeDiff},
$
\left(\fLT^2F\right)(-y) = \left(\fLT T_F\right)(y) = T_{T_F}(y) 
$
where $T_{T_F}(y)$ is the set of all pre-trees constructed from the elements of 
$\,\overline{\cT}_F$ regarded as subobjects. From~(\ref{e:TFk}), and by comparing $T_F$ with $F$ in Theorem~\ref{L:TreeDiff}, these subobjects are:
$\mbox{ a pre-edge: $-\left(T^{(2)}_F\right)^{-1}$; \quad  a $k$-pre-vertex: $T^{(k)}_F$ for $k\ge3.$}$
But $T^{(2)}_F$ is the generating series for the set of all trees $\ft$ in $\cT^\ell$ but with $1$-vertices deleted and their labels retained, no $2$-vertices, at least one edge, and exactly two open ends, which are labelled.  Then $\ft$ is a path with only one edge, which has open ends, and so  $T^{(2)}_F$ corresponds to a negatively weighted anti-edge. It is represented diagrammatically by: \includegraphics{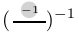}. An object counted by $T^{(k)}_F$ is denoted by \,\,
\includegraphics[]{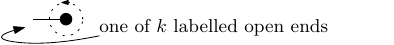}

The glueing of two objects $\fa$ and $\fb$ from $\overline{\cT}_F$ is performed as follows:
\begin{eqnarray*}
\fa \Join \fb 
&:=&
\fa \Join
\includegraphics[]{QFTZfig13.pdf}
\Join \fb 
=
\includegraphics[]{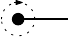}
\Join
\includegraphics[]{QFTZfig13.pdf}
\Join
\includegraphics[]{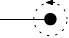} 	\\
&=&
\includegraphics[]{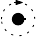}
\Join
\includegraphics[]{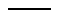}
\Join
\includegraphics[]{QFTZfig13.pdf}
\Join
\includegraphics[]{QFTZfig15edge.pdf}
\Join
\includegraphics[]{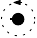}
\\
&=& 
\includegraphics[]{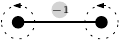}
\quad \in\overline{\cT}_{T_F}.
\end{eqnarray*}
where $\overline{\cT}_{T_F}$ is the set of all pre-trees constructed from $\overline{\cT}_F$ by a finite number of applications of $\Join$.  Edges not weighted multiplicatively by $-1$ are weighted multiplicatively by $+1.$

Let $\overline{\cT}^{\ell\,\pm}$ denote the set of all trees  in $\overline{\cT}_F$, but with each edge weighted multiplicatively by $-1$ or $+1,$ in all possible ways.  Then, under the action of $\Join$ for objects in $\overline{\cT}_F$, we have
\begin{equation}\label{TFbijTpm}
\overline{\cT}_{T_F} \bij \overline{\cT}^{\ell,\pm}_\cF
\end{equation}
since an edge weighted by $-1$ indicates an $F$-edge.

Let $\cB$ be the subset of objects in $\overline{\cT}^{\ell\,\pm}$ with at least one edge. Let
$\Om\colon \cB \rar\cB \colon \fb \rar \fb'$
where $\fb'$ is constructed from $\fb$ by reversing the $\pm$ sign of the multiplicative weight on the edge of $\fb$ incident with the attachment point labelled $1$.  Clearly $\Om$ is a sign-reversing involution. Moreover, $\overline{\cT}^{\ell\,\pm} - \cB$ is the set of all pre-trees with no edges, namely, 
$\left\{\: \Bmyputeps{0.5}{0.45}{preedge}^{-1},\,
\myputeps{-1.00}{0.40}{3Vertex},\, \quad\myputeps{-1.10}{0.40}{4Vertex},\,  \ldots\: \right\}. $
Thus, by the Sign-reversing Involution Lemma,
$$
\gensb{\overline{\cT}^{\ell\,\pm}}{\wt{\sa}\ot\wt{\sfe}\ot\om}(y; u,\lam_2,\lam_3,\ldots)
= -u^{-1} \frac{y^2}{2} + \sum_{k\ge3} \lam_k \frac{y^k}{k!}
$$
so, from~(\ref{TFbijTpm}),
$$
\gensb{\overline{\cT}_{T_F}}{\wt{\sa}\ot\wt{\sfe}\ot\om}(y; u,\lam_2,\lam_3,\ldots)
=-u^{-1} \frac{y^2}{2} +  \sum_{k\ge3} \lam_k \frac{y^k}{k!},
$$
completing the proof.
\end{proof}

\subsubsection{The relationship between connected graphs and edge 2-connected graphs}
Theorem~\ref{L:TreeDiff} and Lemma~\ref{L:qInvLT} may be viewed in a more general way.  Pre-vertices and pre-edges are abstract objects with the property that an attachment point may be attached to (identified with) an open end of a pre-edge.  A pre-graph that cannot be constructed from pre-graphs by means of the $\Join$ operation will be said to be $\Join$-\emph{irreducible}.  The pre-edge, the anti-edge and the pre-vertices will be called the $\Join$-\emph{irreducible constituents}.

In a sense that Theorem~\ref{L:TreeDiff} and Lemma~\ref{L:qInvLT} make precise, we may assert that:
\begin{itemize}
\item [--] $\fLT$ constructs the set of all trees formed from the set of all $\Join$-irreducible constituents.
\end{itemize}

For example, every connected graph may be uniquely decomposed into a tree of edge 2-connected graphs.  There is the following realisations of the abstract pre-vertices and pre-edges:
\begin{itemize}
\item [--] \emph{abstract pre-edges}:  these are the edges of the tree which correspond to cut edges of the graph;
\item [--] \emph{abstract pre-vertices}:  these are the edge 2-connected subgraphs, together with their points of attachment with edges of the tree as attachment points (see Example~\ref{ex:2edgeconnected}).
\end{itemize}

With this in mind, we make the following definition:

\begin{definition}[The set $\Gtwolc$]\label{D:GaSeries}
$\Gtwolc$ is the set of all $2$-edge connected graphs in $\cG^{\ell,c}$ with at least one point of attachment.
\end{definition}

The following result is now immediate from Theorem~\ref{L:TreeDiff} and the above identification of pre-edges and pre-vertices. 

\begin{theorem}[$\Gtwolc$ and $\Ga_F$]\label{Th:FolkTH}
Let  
$\Ga_F(z; u, \lam_2, \lam_3, \ldots) := \gensb{\Gtwolc}{\wt{\sa}\ot\wt{\sfe}\ot\om}(z; u,\lam_2,\lam_3,\ldots).$
Then
$$(\fLT W_F)(z) =  - u^{-1}\frac{z^2}{2!} + \Ga_F(z).$$
\end{theorem}

\newcommand{\cFb}{\overline{\cF}}
\newcommand{\seb}{\overline{\sfe}}
\newcommand{\svb}{\overline{\sv}}
\newcommand{\lamb}{\overline{\lam}}

\begin{proof}
Let $\widehat{\Ga}(z) := -u^{-1} \frac{z^2}{2!} + \Ga_F(z).$

\myss{Claim 1}
$\widehat{\Ga}(z) = \gensb{\cFb}{\wt{\sa}\ot\wt{\seb}\ot\wt{\svb_1}\ot\wt{\svb_2}\ot \cdots}
(z; -u, \lamb_1, \lamb_2, \ldots)$
where
$$
\cFb := \{   {\seb}^{-1}, \svb_1, \svb_2, \ldots  \}
\qquad\mbox{and}\qquad
\svb_k := \biguplus_{\substack{\fg\in \Gtwolc, \\ \wt{\sa}(\fg)=k.}}\{ \fg \}.
$$

\myss{Proof of Claim 1}
We note that
$\Gtwolc = \biguplus_{k\ge1} \svb_k$
and
$u^{-1}\frac{z^2}{2!} = \gensb{\{ {\seb}^{-1} \}} {\wt{\sa}\ot \wt{\seb}} (z; -u).$
Claim~1 follows. \hfill $|$

We show that there is a bijection between $\cG_\cF^{\ell,c}$ that preserves $\wt{\sv_1}$ (labelled $1$-vertices). The other arguments are substituted by means of a compositional argument.

\myss{Claim 2}
$
W_F(y) = \gensb{\cT^\ell_{\cFb}}{\wt{\sv_1} \ot \wt{\seb} \ot \wt{\svb_1} \ot  \wt{\svb_2} \ot \cdots} (y; u, \lamb_1, \lamb_2, \ldots).
$

\myss{Proof of Claim 2}
A connected graph $\fg$ with labelled $1$-vertices can be represented uniquely as a tree of edge 2-connected graphs: the tree-edges are cut-edges and and the tree-vertices are:
\begin{itemize}
\item labelled $1$-vertices;
\item Edge $2$-connected subgraphs of $\fg$, together with the $k\ge1$ points of attachment to the rest of $\fg.$
\end{itemize}
Thus
\begin{equation}\label{e:GHTG2conn}
\cG_\cF^{\ell,c} \bij \cT^\ell_{\cFb} \circ \cG^{\ell, \mathsf{2-conn}}_\cF
\end{equation}
where the multivariate composition `$\circ$' is defined by
$$
\svb_k \mapsto 
\biguplus_{\substack{\fg\in \cG^{\ell, \mathsf{2-conn}}_\cF, \\ \wt{\sa}(\fg)=k.}}\{ \fg \},
\qquad\mbox{where}\quad k\ge1
$$
and ``$\myputeps{1}{0.45}{preedge}$'' maps to  cut-edges.
The composition is well defined since attachment points of graphs in $\cG_\cF^{\ell,c}$, and
in particular those connected to cut edges, are labelled.
Claim~2 follows by taking the generating series for each side and using the multivatiate composition lemma (the bijection is weight-preserving in $\wt{\sv_1}$).  \hfill $|$

From Theorem~\ref{L:TreeDiff} we have
$\left(\fLT \widehat{\Ga}_F\right) (y) =  W_F (-y)$
since, from Claim~1 it follows that $\widehat{\Ga}_F(z) = \gensb{\cFb}{\wt{\sa} \ot \wt{\seb}}$,
and from Claim~2 that $W_F = T_{\overline{F}}.$ Then applying $\fLT$ to each side gives
$
\left( \fLT^2 \widehat{\Ga}_F\right) (-y) = \fLT\left(W_F(y)\right)$., so from Theorem~\ref{T:QuasiInvolLT},
$\widehat{\Ga}_F (y) = \fLT \left(W_F\right)(y),$ completing the proof.
\end{proof}

Since the proof of~(\ref{CD:diag12}) has already been completed, this now completes the proof of Diagram~(\ref{CD:diagr2}).

\subsection{Relationship between trees of connected graphs and connected graphs}
In view of the combinatorial constructions given in Section~\ref{SS:QIpLT} and Section~\ref{SS:QIpFT} for $\fLT$ and $\fFT$, the combinatorial Legendre and Fourier transforms, respectively, we now seek to relate the two constructions.

\begin{lemma}
$\left.\cG^\ell_{\cW_\cF}\right|_{\cT^\ell_{\cW_\cF}}
\bij \ol{\cG_\cF^{\ell,\mathsf{conn} , \pm}}.
$
\end{lemma} 
\begin{proof}
From Theorem~\ref{Th:FolkTH}, we have
$(\fLT W_F)(-z) =  u^{-1}\frac{z^2}{2!} - \Ga_F(z),$ 
which we proved by constructing a bijection
\begin{equation}\label{e:LTcomp}
\cT^\ell_{\cW_\cF} \bij \ol{\cG_\cF^{\ell,\mathsf{conn} , \pm}}
\qquad\mbox{where $\pm$ weights the cut-edges (Legendre case)}
\end{equation}
and then used a sign-reversing involution on $\ol{\cG_\cF^{\ell,\mathsf{conn} , \pm}}$ to remove all \emph{cut-edges}, thereby obtaining
$\ds{u^{-1}\frac{z^2}{2!} - \Ga_F(z)}.$
On the other hand, from Theorem~\ref{T:QuasiInvFT} we have
$\left(\fFT e^{\cW_\cF}\right) (z) = e^{F(z)},$
which we proved by finding a bijection
\begin{equation}\label{e:FTcomp}
\cG^\ell_{\cW_\cF}\bij  \ol{\cG_\cF^{\ell , \pm}}
\qquad\mbox{where $\pm$ weights internal edges (Fourier case)}
\end{equation}
and then used a sign-reversing involution on $ \ol{\cG_\cF^{\ell , \pm}}$ to remove all \emph{internal edges}, and we obtain $e^{F(z)}$.
We observe that (\ref{e:LTcomp}) restricts to~(\ref{e:FTcomp}) under the injection
$
\cT^\ell_{\cW_\cF} \rar \ol{\cG_\cF^{\ell,\mathsf{conn} , \pm}}.
$
See Example~\ref{ex:2edgeconnected}.

\begin{example}\label{ex:2edgeconnected}
We illustrate the injection between $\cT^\ell_{\cW_\cF}$ and $\ol{\cG_\cF^{\ell,\mathsf{conn} , \pm}}$ with the following example:
\phantom{Figure goes here}
\begin{center}
    \includegraphics[]{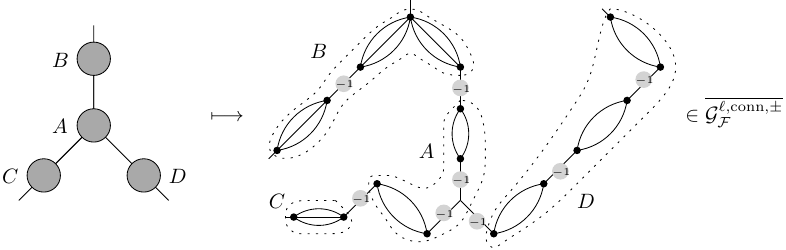}
\end{center}
\end{example}

The restriction is certainly injective.  It is also surjective since, given a graph 
$\fg\in\ol{\cG_\cF^{\ell,\mathsf{conn} , \pm}}$, we may view $\fg$ as a tree of edge 2-connected graphs.  Some of the edges of $\fg$ will be weighted by $-1.$

From the proof of Theorem~\ref{Th:FolkTH}, we have
$
\cG_\cF^{\ell,\mathsf{conn}} \bij \cT^\ell_{\cFb} \circ \cG^{\ell, \mathsf{2-conn}}_\cF,
$
stating that a connected graph is a tree of 2-connected graphs.  Thus
$\ol{\cG_\cF^{\ell,\mathsf{conn},\pm}} 
\bij
\ol{ \cT^{\ell,\pm}_{\cFb}} \circ \cG^{\ell, \mathsf{2-conn}}_\cF$
where
 \begin{itemize}
 \item [--]  on the left, the cut-edges are weighted $\pm1$;
 \item [--]  on the right, the internal edges are weighted $\pm1$.
 \end{itemize}
 But there is also a bijection
 $
 \ol{ \cT^{\ell,\pm}_{\cFb}} \circ \cG^{\ell, \mathsf{2-conn}}_\cF
 \bij 
 \cT^\ell_{\cW_\cF}.
 $
 Therefore,
$\ol{\cG_\cF^{\ell,\mathsf{conn},\pm}} 
\bij
 \cT^\ell_{\cW_\cF}.
$
Thus the restriction is surjective.
\end{proof}

\newcommand{\gr}{\mathbf{\fg}}

\section{Outlook}
\label{S:outlook}

We constructed and investigated algebraic and combinatorial formulations of the Fourier and Legendre transform. These are all-purpose methods that can be used also beyond quantum field theory to consistently transform formal power series, with the key structural equations of the Legendre and Fourier transforms still in place, with the advantage that one no longer needs to assume the convergence of these power series.

\subsection{Complexity of transforms}
For our transforms, which are quasi-involutive maps of formal power series, it would be interesting to study their computational complexity and to investigate if they preserve certain classes of power series, such as {\em algebraic, differentiably finite} or satisfying an {\em algebraic differential equation} \cite[Ch. 6]{EC2}. In this context, see also, for example, \cite[Remark 7]{JKM1} and the formulas in Appendix~\ref{S:appendix_explicit_examples}.

\subsection{Borinsky's approach for asymptotics of power series of Feynman diagrams} \label{subsec: Borinsky}

In \cite{Bo}, the author studied the asymptotic expansions of the power series of Feynman graphs in the univariate case, also called $0$-dimensional QFT. As part of the analysis, the author also views functions like $Z_F(y), W_F(y), \Gamma_F(z)$ from Sections~\ref{S:FormalFT1} and \ref{S:FormalLT1} as formal power series counting Feynman graphs. The asymptotics of these functions are then obtained algebraically and explicit asymptotics are obtained for special cases like $\phi^3$ and $\phi^4$ theories. In Appendix~\ref{S:appendix_explicit_examples} we give explicit expressions for $Z_F(y)$ and the Legendre transform of $F(y)$ and compare them with expressions from \cite{Bo}.

While \cite{Bo} is concerned with asymptotics of the power series of Feynman diagrams, we are here, in contrast, concerned with generalizing the relations involving the functions $Z_F(y), W_F(y), \Gamma_F(z)$ into algebraic and combinatorial Fourier and Legendre transformation in the ring of formal power series, as we motivated in the Introduction.

\subsection{Extending transforms to multivariate power series}

Our results also suggest further generalization to the multivariate case where we have $0$-dimensional QFT of several fields. To this end, the algebraic/combinatorial Fourier transform (see Theorems~\ref{thm:fFTa} and \ref{T:gsZF}) can be extended to more variables as follows. Let $F_{int}({\bf x}) := \sum_{\bf k} \lambda_{\bf k} \frac{x_1^{k_1} x_2^{k_2} \cdots}{k_1!k_2!\cdots}$, where ${\bf x}=(x_1,x_2,\ldots,x_n)$, ${\bf k}=(k_1,k_2,\ldots,k_n)$ then
\begin{equation} \label{Z: multivariate}
Z_F({\bf y}) := \exp(F_{int}({\bf x})) \,{\Big \vert}_{x_i \mapsto \partial_{y_i}}\, \exp\Big(\sum_i u_{i} \frac{y_i^2}{2}\Big),
\end{equation}
would give the generating series of graphs where the attachment points and $1$-vertices now have $n$ types. 

Unlike the Fourier transform, the Legendre transform is nonlinear. Therefore, a combinatorial multivariate version of the Legendre transform is more subtle  (see \cite[Remark 2]{JKM1}). Such a multivariate analogue is possible and will involve the compositional inverse of multivariate power series and the multivariate Lagrange inversion theorem \cite[\S 1.1.4]{GJ}.

Other generalizations could be including fermionic variables and extending  beyond $0$-dimensional QFT. To extend this formalism to multidimensional spacetime will involve describing all the coefficients, $\lambda_i$ and $u$, in the action as operators acting on a function space of fields.

\subsection{Anomalies and quantum information} 
Further, when applied to quantum field theory, the new algebraic and combinatorial transforms should also allow one to gain new perspectives and insights into key properties of perturbative quantum field theory that are usually thought of in analytic terms. For example, we already showed that the defining equation of the  Legendre transform reduces to the Euler characteristic for tree graphs \eqref{e:PtnUnityA}, which is an entirely non-analytic and therefore robust statement. It should also be very interesting to explore, for example, the algebraic and combinatorial origin 
of anomalies, since they are normally derived from the nontrivial analytic properties of the path integral measure.

Finally, the fundamental discreteness of the purely algebraic and combinatorial  transforms invites an information-theoretic analysis, which might eventually also help in reformulating quantum field theory in information-theoretic terms. 
$$$$
\bf Acknowledgements: \rm We thank Michael Borinsky, Ali Assem Mahmoud, Karen Yeats, and the anonymous referees for helpful comments and suggestions.  A.K. acknowledges support through the Discovery Program of the National Science and Engineering Research Council of Canada (NSERC) and through two Google Faculty Research Awards. D.M.J. was partially supported by an NSERC GRANT. A.H.M. was partially supported by the NSF Grant DMS-1855536.

\medskip
\bibliographystyle{plain}


\appendix

\section{Corollaries and examples of the combinatorial Fourier transform}\label{S:appendix_explicit_examples}

In this appendix we give examples of the combinatorial Fourier transform. We start by giving explicit expansions for this transform in terms of edges, loops ($\#\text{edges} - \#\text{vertices}$) and vertices.

\begin{corollary}[edge series]
Let $F(x) = -u^{-1}x^2/2 + \sum_{k\geq 3} \lambda_k x^k/k!$ and $Z(y)$ be its Fourier transform, then $Z(y) = \sum_{k,m\geq 0} Z^{(k,m)} u^m y^k/k!$ where
\begin{equation} \label{eq:master_edge}
    Z^{(k,m)} = (2m-1)!!\sum_{n_3,n_4,\ldots} \prod_{j\geq 3} \frac{1}{n_j!}\left( \frac{\lambda_j}{j!}\right)^{n_j} \hbar^{1+m-\sum_{j\geq 3} n_j-k},
\end{equation}
where the finite sum is over all finite  tuples $(n_3,n_4,\ldots)$ of nonnegative integers satisfying the Euler relation $\sum_{j\geq 3}^{2m-k} j\cdot n_j = 2m-k$. 
\end{corollary}

\begin{proof}
By Theorem~\ref{T:gsZF} we have that 
\begin{equation} \label{eq:defZ_pf_master_eq_Z}
Z(y) := \left. e^{ \frac{1}{\hbar}\Fi (x)}\right|_{x\mapsto \partial_y}  e^{\frac 12 \hbar u y^2}
\end{equation}
equals the generating function of graphs
\begin{equation} \label{eq:pf_master_eq_Z}
Z(y) = \gensb{\cG^\ell}{\om_\sa\ot \om_{\sv_1} \ot \om_{\sfe} \ot \om} (x=1, y; u, \lam_2, \lam_3, \ldots).
\end{equation}
Then the coefficients $Z^{(k,m)}$ give the generating polynomial for graphs with $k$ $1$-vertices and $m$ edges. Expanding the defining formula \eqref{eq:defZ_pf_master_eq_Z} for $Z(y)$ and collecting these coefficients gives the desired expression for $Z^{(k,m)}$. 
\end{proof}

\begin{corollary}[$\hbar$ series]
Let $F(x) = -u^{-1}x^2/2 + \sum_{k\geq 3} \lambda_k x^k/k!$ and $Z(y)$ be its Fourier transform, then $Z(y) = \sum_{k \geq 0, n\geq -1} Z_{(k,n)} \hbar^n y^k/k!$ where
\begin{equation} \label{eq:masterh}
Z_{(k,n)} = \sum_{n_3,n_4,\ldots} (2n+2k+2{\textstyle \sum_{j\geq 3}} n_j -1)!! \,\, \prod_{j\geq 3} \frac{1}{n_j!}\left( \frac{\lambda_j}{j!}\right)^{n_j} u^{n+k+\sum_{j\geq 3} n_j},    
\end{equation}
where the finite sum is over all finite  tuples $(n_3,n_4,\ldots)$ of nonnegative integers satisfying the Euler relation $\sum_{j\geq 3}^{2n+k-2} (j-2)n_j = 2n+k$. 
\end{corollary}

\begin{proof}
From \eqref{eq:pf_master_eq_Z}, the coefficients $Z_{(k,n)}$ give the generating polynomial for graphs with $k$ $1$-vertices and $n$ loops. Such graphs have $n+k+\sum_{j} n_j$ edges, where $n_j$ is the number of internal $j$-vertices of the graph. Expanding the defining formula \eqref{eq:defZ_pf_master_eq_Z} for $Z(y)$ and collecting these coefficients gives the desired expression for $Z^{(k,m)}$. 
\end{proof}

Often, the action only contains one vertex and the perturbative expansion of $Z(y)$ is given with respect to the power of $\lambda$ and $y$. We include an explicit expression of this expansion.

\begin{corollary}[$\lambda$ series]
Let $F(x) = -u^{-1}x^2/2 + \lambda_d x^d/d!$ and $Z(y)$ be its Fourier transform, then 
\begin{equation} \label{eq:masterlambda}
Z(y) = \sum_{k\geq 0, v\geq 0} (d v + k -1)!!  \frac{1}{v!} \left(\frac{\lambda_d}{d!}\right)^v \,\,\hbar^{((d-2)v-k)/2}u^{(dv+k)/2} \frac{y^k}{k!}.    
\end{equation}
\end{corollary}

We compare these explicit expansions with the expansion of the Legendre transform from \cite{JKM1}.

\begin{corollary}[{\cite[Cor. 3]{JKM1}}]
Let $F(x)= - x^2/2 + \sum_{n\geq 3} \lambda_n x^n/n!$ be a formal power series and $T(y)$ be its Legendre transform then $T(y)=-y^2/2 + \sum_{n\geq 3} T^{(n)} y^n/n!$ where $T^{(n)}$ for $n\geq 3$ is 
\begin{equation} \label{eq:degist}
T^{(n)} = \sum_{n_3,n_4,\ldots,n_k} (n-2+\sum_{j=3}^k n_j)! \prod_{j= 3}^k \frac{1}{n_j!}\left(\frac{\lambda_j}{(j-1)!}\right)^{n_j},
\end{equation}
where the sum is over all finite tuples $(n_3,\ldots,n_k)$ of nonnegative integers satisfying the Euler relation $\sum_{j=3}^k (j-2)n_j = n-2$.
\end{corollary}

\begin{proposition}
The series of trees and series of graphs are related by
\begin{equation} \label{eq:graphs2trees}
T(y) = [\hbar^{-1}] \, \log Z(y)
\end{equation}
\end{proposition}

\begin{proof}
The series of connected graphs is given by $\log Z(y)$. To recover the trees we calculate the coefficient of $\hbar^{-1}$ from the series of connected graphs.
\end{proof}

Next, we compare our explicit expressions of $Z(y)$ with the expression of Borinsky \cite{Bo} for $Z(0)$ and $Z(y)$ in \eqref{eq:defZ_pf_master_eq_Z}.

\begin{proposition}[{Borinsky \cite[Prop. 2.2]{Bo}}]
Let $F(x) = -x^2/2 + \sum_{k\geq 3} \lambda_k x^k/k!$ and $Z_F(y)$ be its Fourier transform then
\[
Z_F(0) = \sum_{n\geq 0} \hbar^n (2n-1)!!\, [x^{2n}]\, x(y),
\]
where $x(y)$ is the unique solution of $y=\sqrt{-2F(x(y))}$, using the positive branch of the square root.
\end{proposition}

\begin{proposition}[{Borinsky \cite[Sec. 4]{Bo}}]
Let $F(x) = -x^2/2 + \sum_{k\geq 3} \lambda_k x^k/k!$ and $Z_F(y)$ be its Fourier transform then
\[
Z_F(y) = \exp\left(-\frac{1}{\hbar}\left(\frac{x_0^2}{2} + \Fi(x_0)+x_0 y\right)\right) \cdot Z_{G}(0), 
\]
for $G(y)=-x^2/2 + \Fi(x+x_0) - \Fi(x_0)-x \Fi'(x_0)$ and $x_0=x_0(y)$ is the unique power series solution to $x_0(y)=\Fi'(x_0(y))+y$.
\end{proposition}

As mentioned in the Introduction, our new methods allow one to define the algebraic Fourier transform, $Z(y)$, independently of whether or not the action is convex or convergent. In particular, we now give explicit expressions for the algebraic Fourier transform, $Z(y)$, for several actions that are non-convex or possess a finite or vanishing radius of convergence whose algebraic Legendre transform we previously calculated in \cite{JKM1}. For other interesting examples see Borinsky \cite{Bo}. 

\begin{example}[tetravalent graphs - convex polynomial action] \label{ex:tetraconv}
The action $F_1(x) = -x^2/2 - x^4/4!$ (i.e. $F^{(4)}=-1$ and $F^{(k)}=0$ for other $k$) is a convex function. 
Using \eqref{eq:masterh} one obtains
\[
Z(y) = \sum_{j \geq 0} \frac{y^{2j}}{(2j)!} \sum_{n \geq -1} \hbar^n \frac{(-1)^{n+j}(4n+6j-1)!!}{24^{n+j}(n+j)!} 
\]
\end{example}

\begin{example}[tetravalent graphs - non convex action] 
\label{ex:tetranonconv} 
In contrast to Example~\ref{ex:tetraconv}, the action $F_2(x) = -x^2/2 + x^4/4!$ (i.e. $F^{(4)}=1$ and $F^{(k)}=0$ for other $k$) is not a convex function.
Using \eqref{eq:masterh} one obtains
\[
Z(y) = \sum_{j \geq 0} \frac{y^{2j}}{(2j)!} \sum_{n \geq -1} \hbar^n \frac{(4n+6j-1)!!}{24^{n+j}(n+j)!} 
\]
\end{example}

\begin{example}[tetravalent graphs - polynomial action with a parameter] 
\label{ex:tetraparam} 
In contrast to Example~\ref{ex:tetraconv}, the action $F_3(x) = -x^2/2 + \lambda_4 x^4/4!$ (i.e. $F^{(4)}=\lambda_4$ and $F^{(k)}=0$ for other $k$) is not a convex function.
Using \eqref{eq:masterh} one obtains
\[
Z(y) = \sum_{j \geq 0} \frac{y^{2j}}{(2j)!} \sum_{n \geq -1} \hbar^n \lambda_4^{n+j}\frac{(4n+6j-1)!!}{24^{n+j}(n+j)!} 
\]
\end{example}

\begin{example}[trivalent graphs - non convex polynomial action] \label{ex:trinonconv}
The action $F_4(x) = -x^2/2 + x^3/3!$ (i.e. $F^{(3)}=1$ and $F^{(k)}=0$ for other $k$) is a convex function. Using \eqref{eq:masterh} one obtains
\[
Z(y) = \sum_{k \geq 0} \frac{y^k}{k!} \sum_{n \geq -1} \hbar^n \frac{(6n+4k-1)!!}{6^{2n+k}(2n+k)!}. 
\]
Borinsky \cite[Sec. 6.2]{Bo} derived the following expression for $Z(y)$ as a composition of the constant term in $y$ of this same series: 
\[
Z(y) = \frac{\exp(((1-2y)^{3/2} -1 +3y)/3\hbar)}{(1-2y)^{1/4}} Z_0\left(\frac{\hbar}{(1-2y)^{3/2}}\right), 
\]
where $Z_0(\hbar):= [y^0] Z(y) = \sum_{n=0}^{\infty} \hbar^n \frac{(6n-1)!!}{6^{2n} (2n)!}$.
\end{example}

\begin{example}[Nonconvex non-polynomial action convergent on $\mathbb{R}$]
Let us now consider an example of a non-polynomial action, i.e., an action which possesses infinitely many different types of vertices. The graphs for the action $F_5(x) = -x^2/2 + \sum_{k=3}^{\infty} x^k/k!=e^x-1-x-x^2$, {\em i.e.} $F^{(k)}=1$ for $k=3,\ldots$.  Using \eqref{eq:masterh} one obtains
\[
Z(y) = \sum_{k\geq 0,n \geq -1} \frac{y^k}{k!} \hbar^n \sum_{n_3,n_4,\ldots} (2n+2k+2{\textstyle \sum_{j\geq 3}} n_j -1)!!  \prod_{j\geq 3} \frac{1}{n_j!}\left( \frac{1}{j!}\right)^{n_j},
\]
where the sum is over all finite  tuples $(n_3,n_4,\ldots)$ of nonnegative integers satisfying the Euler relation $\sum_{j\geq 3} (j-2)n_j = 2n+k$. 
\end{example}

\begin{example}[Nonconvex non-polynomial action, radius of convergence $1$]
Let us consider another example of a non-polynomial action with radius of convergence $1$. The graphs for the action $F_6(x) = -x^2/2 + \sum_{k=3}^{\infty} x^k/k=-\ln(1-x) -x- x^2$, {\em i.e.} $F_6^{(k)}=(k-1)!$ for $k=3,\ldots$. Using \eqref{eq:masterh} one obtains
\[
Z(y) = \sum_{k\geq 0,n \geq -1} \frac{y^k}{k!} \hbar^n \sum_{n_3,n_4,\ldots} (2n+2k+2{\textstyle \sum_{j\geq 3}} n_j -1)!!  \prod_{j\geq 3} \frac{1}{j^{n_j} n_j!},
\]
where the sum is over all finite  tuples $(n_3,n_4,\ldots)$ of nonnegative integers satisfying the Euler relation $\sum_{j\geq 3} (j-2)n_j = 2n+k$. 
\end{example}

\begin{example}[Action with zero radius of convergence]
Lastly, let us consider a non-polynomial action with zero radius of convergence. The graphs for the action $F_7(x) = -x^2/2 + \sum_{k\geq 3} k!x^k$, {\em i.e.} $F^{(k)}=(k!)^2$ for all $k=3,\ldots$. Using \eqref{eq:masterh} one obtains
\[
Z(y) = \sum_{k\geq 0,n \geq -1} \frac{y^k}{k!} \hbar^n \sum_{n_3,n_4,\ldots} (2n+2k+2{\textstyle \sum_{j\geq 3}} n_j -1)!!  \prod_{j\geq 3} \frac{j!^{n_j}}{n_j!},
\]
where the sum is over all finite  tuples $(n_3,n_4,\ldots)$ of nonnegative integers satisfying the Euler relation $\sum_{j\geq 3} (j-2)n_j = 2n+k$. 
\end{example}

\section{A possible connection to tempered distributions} \label{sec:exceptionalelement}

\subsection{A series expansion}  
We consider the generalisation $e^{\ha\mP\, \py^2} e^{\ha\mQ\, y^2}$ of the expression $e^{\ha a \py^2}  \, e^{\ha a^{-1} y^2}$ that appears in Theorem~\ref{thm:fFTa} to determine the conditions under which it exists as a formal power series.

\begin{lemma}\label{L:PQexp}
Let $\mP$ and $\mQ$ be indeterminates such that  $\mP\mQ\neq 1.$ Then
$$ 
e^{\ha\mP\, \py^2} e^{\ha\mQ\, y^2}  
= \frac{1}{\sqrt{1-\mP\mQ}} \cdot e^{\ha(1-\mP\mQ)^{-1}\mQ\, y^2} 
\in \bbC[[p,q]] \,[[y]].
$$
\end{lemma}
\begin{proof}[Algebraic proof] 
It will be convenient in the proof to replace $\mP$ by $2\mP$ and $\mQ$ by $2\mQ$. Clearly
\begin{align*}
e^{\mP\, \py^2} e^{\mQ\, y^2} 
&= \sum_{j\ge i\ge 0} \frac{(2j)_{2i}}{i!j!}  \mP^i \mQ^j y^{2(j-i)}
= \sum_{r\ge0} \left(\frac{y^2}{\mP}\right)^r B_r  
\intertext{where}
B_r &\coloneqq \sum_{j\ge r} \frac{(2j)_{2j-2r}}{{j!}(j-r)!} (\mP\mQ)^j
= \frac{(\mP\mQ)^r}{(2r)!} \sum_{s\ge0}  \frac{(2r+2s)!}{s!(r+s)!} (\mP\mQ)^s. 
\end{align*}
Observing that 
$(2r)! = (-1)^r 4^r r! (-\sha)_r \quad\mbox{and}\quad (-\sha)_{r+s} = (-\sha)_r (-r-\sha)_s,$
we have
\begin{equation*}
 \frac{(2r+2s)!}{(2r)!s!(r+s)!} 
= \frac{(-4)^s}{r! s!} \frac{(-\ha)_{r+s}} {(-\ha)_{r}\,\,\,\,\,\,}
=\frac{(-4)^s}{r!} \binom{-r-\ha}{s}
\end{equation*}
so 
$B_r = \frac{1}{r!} (\mP\mQ)^r (1-4\mP\mQ)^{-r-\ha},$
provided $4\mP\mQ \neq 1.$
The result now follows. 
\end{proof} 

\begin{proof}[Combinatorial proof]
	We begin by seeking a combinatorial interpretation of the series
\begin{equation}\label{B1Left}
 e^{\frac{1}{2} p\partial_y^2}  e^{\frac{1}{2} qy^2}.
 \end{equation}
	
	A \emph{perfect matching} is a graph on a set of $2n$ labelled vertices $1, 2, \ldots, 2n$, each of degree $1$ and having $n$ edges.  The generating series for all such graphs is $e^{qy^2/2}$,   This is exponential in $y$, which marks vertices, and ordinary in $q$ which marks edges.
	Now consider the action of the formal differential operator $e^{p\partial_y^2/2}$, where $p$ is an indeterminate, on the perfect matchings.   The operator $p\partial_y^2/2$ selects a pair of vertices of the matchings and joins them with an edge and attaches an indeterminant $q$ to the edge.   The operator $e^{p\partial_y^2/2}$ carries this out in all possible ways.  This gives the generating series for a set of graphs with components that are either \\

\begin{itemize}
    \item[1.]  a path whose end vertices are labelled and interior vertices are unlabelled, and  whose product of edge markers is $(qp)^kq$ where $k$ is a positive integer, or 
\item[2.] a cycle whose product of  its $2k$ edge markers is $(pq)^k$.\\ 
\end{itemize}

In Case~1, the contribution to the generating series for each such component is
$$
\sum_{k\ge0}  (qp)^k q \frac{y^2}{2!} = (1-qp)^{-1}q \frac{y^2}{2!},
$$
 so the generating series for the disjoint union of these is
 $$e^{(1-qp)^{-1} q y^2/2}.$$

In Case~2, the contribution from the edge markers is $\sum_{k\ge0} (qp)^k = (1-qp)^{-1}$ for $pq\neq 1$.  These appear in disjoint cycles so, for each cycle the contribution is $\frac{1}{2} \log (1-pq)^{-1}$, where the factor of $1/2$ accounts the cycles are not directed.  The generating series for the disjoint union of all such cycles  is
$$
\exp \left( \frac{1}{2} \log\left(1-pq)^{-1}\right)\right) = (1-pq)^{-1/2}.
$$
The generating series we seek is the product of the generating series for these two cases,  which gives the generating series
\begin{equation} \label{eB1}
\frac{1}{\sqrt{1-pq}} e^{\frac{1}{2}(1-pq)^{-1} q y^2}
\end{equation}
The result follow by equating (\ref{B1Left}) and (\ref{eB1}).
\end{proof}

\begin{remark}
This result may be extended to the case in which $p$ and $q$ are Laurent series $p(a)$ and $q(a)$ with a finite number of terms with negative exponents in the indeterminate $a$.
\end{remark}

\subsection{The exceptional element} 
Let
\begin{equation*} 
 \PQ_{p,q}(y) \coloneqq 
e^{\ha\mP\, \py^2} e^{\ha\mQ\, y^2}.
\end{equation*}
By Lemma~\ref{L:PQexp}, $\PQ_{\mP,\mQ} (y)$ is not a formal power series in $y$ if $\mP\mQ=1$. Nevertheless, let us denote the excluded object $\PQ_{a,a^{-1}}$ by $\PQ_a$.  It will be referred to as the \emph{exceptional element}. When acted upon by $e^{-\ha c\,\py^2}$ ($c$ an indeterminate),  Lemma~\ref{L:PQexp} shows that $\PQ_a$ has the \emph{primary property} that
\begin{equation*} 
e^{-\ha c\, \py^2} \;  \PQ_a  =  \sqrt{\frac{a}{a-c}} \, \cdot e^{\ha \, c^{-1}y^2},
\end{equation*}
is a formal power series in $y,$ \emph{via} $e^{-\ha c\, \py^2} \; \PQ_a  =  \PQ_{a-c, a^{-1}}(y)$.
That is, $\PQ_a$ is a valid expression in the appropriate ring of formal power series when it is acted upon by a differential operator of the form $e^{\ha c\, \py^2}$.
We note that when $\Delta_{p,q}$ is viewed analytically as an operator on a space of functions instead of an operator on formal power series, then it is represented by the tempered distribution $\de(x)$, where the latter is defined implicitly by the property
$$\int_\bbR f(x) \de(x) \, \mrd x \coloneqq f(0).$$
Indeed, we have showed in \cite[Eq. (7)]{JKM2} that
\[
\de(x) =  \frac{1}{\sqrt{2\pi a}} \, e^{-\ha a\px^2} e^{-\ha a^{-1} x^2} = \frac{1}{\sqrt{2\pi a}} \Delta_{-a}(x).
\]

\end{document}